\pdfoutput=1
\documentclass[
final
]{dmtcs-episciences}

\usepackage[utf8]{inputenc}
\usepackage{subfigure}

\usepackage[round]{natbib}

\usepackage{tikz}
\usetikzlibrary{arrows,chains,matrix,positioning,scopes,calc,automata}
\usepackage{amsmath}
\usepackage{amsthm}
\usepackage{amssymb}
\usepackage{amsbsy}

\newtheorem{theorem}{Theorem}
\newtheorem{lemma}[theorem]{Lemma}
\newtheorem{corollary}[theorem]{Corollary}
\theoremstyle{plain}
\newtheorem{proposition}[theorem]{Proposition}
\theoremstyle{definition}
\newtheorem{definition}[theorem]{Definition}
\newtheorem{example}[theorem]{Example}

\newcommand{\aryrelation}[1]{\Sigma^*_1\times\ldots\times\Sigma^*_{#1}}
\newcommand{\aryomegarelation}[1]{\Sigma^\omega_1\times\ldots\times\Sigma^\omega_{#1}}
\newcommand{\maybeomega}{$(\omega\text{-})$}
\newcommand{\omegaintitle}{\texorpdfstring{$\mathbb{\omega}$}{\text{omega}}}

\DeclareMathOperator{\lcm}{lcm}
\DeclareMathOperator{\dom}{dom}
\DeclareMathOperator{\equiindex}{index}
\DeclareMathOperator{\rev}{rev}

\author{Christof L\"oding\affiliationmark{1} \and Christopher Spinrath\affiliationmark{2}}

\title[Decision Problems for Subclasses of Rational Relations]{Decision Problems for Subclasses of Rational Relations over Finite and Infinite Words}
\affiliation{
	RWTH Aachen University, Germany\\
	TU Dortmund University, Germany}

\keywords{rational relations, automatic relations, omega-automata, finite transducers, visibly pushdown automata}
\received{2018-3-20}
\revised{2018-12-17, 2019-1-11}
\accepted{2019-1-17}

\begin{document}

\publicationdetails{21}{2019}{3}{4}{4393}

\maketitle

\begin{abstract}
	We consider decision problems for relations over finite and infinite words defined by finite automata. We prove that the equivalence problem for binary deterministic rational relations over infinite words is undecidable in contrast to the case of finite words, where the problem is decidable. Furthermore, we show that it is decidable in doubly exponential time for an automatic relation over infinite words whether it is a recognizable relation. We also revisit this problem in the context of finite words and improve the complexity of the decision procedure to single exponential time. The procedure is based on a polynomial time regularity test for deterministic visibly pushdown automata, which is a result of independent interest.
\end{abstract}

\section{Introduction} \label{sec:intro}
We consider in this paper algorithmic problems for relations over
words that are defined by finite automata. Relations over words extend
the classical notion of formal languages. However, there are different
ways of extending the concept of regular language and finite automaton
to the setting of relations. Instead of processing a single input
word, an automaton for relations has to read a tuple of input
words. The existing finite automaton models differ in the way how the
components can interact while being read. In the following, we briefly
sketch the four main classes of automaton definable relations, and
then describe our contributions.

A (nondeterministic) finite transducer (see, \textit{e.g.},
\cite{berstel79,Sak09}) has a standard finite state control and at
each time of a computation, a transition can consume the next input
symbol from any of the components without restriction (equivalently,
one can label the transitions of a transducer with tuples of finite
words).  The class of relations that are definable by finite
transducers is referred to as the class of rational relations.
In the binary case, the first tape is often referred to as input, and the second one as output tape.
The class of rational relations is not
closed under intersection and complement, and many algorithmic
problems, like universality, equivalence, intersection emptiness, are
undecidable (for details we refer to \cite{RS59}).  A deterministic version of finite
transducers defines the class of deterministic rational relations (see
\cite{Sak09}) with slightly better properties compared to the
nondeterministic version, in particular it has been shown by \cite{Bir73,HarjuK91}
that the equivalence problem is decidable.

Another important subclass of rational relations are the synchronized
rational relations which have been studied by \cite{FS93} and are defined
by automata that
synchronously read all components in parallel (using a padding symbol
for words of different length). These relations are often referred to
as automatic relations, a terminology that we also adopt, and
basically have all the good properties of regular languages because
synchronous transducers can be viewed as standard finite automata over
a product alphabet. These properties lead to applications of automatic
relations in algorithmic model theory as a finite way of representing
infinite structures with decidable logical theories (so called
automatic structures, \textit{cf.}\
\cite{KhoussainovN95,BlumensathG00}), and in regular model checking, a
verification technique for infinite state systems (\textit{cf.}\
\cite{Abdulla12}).

Finally, there is the model of recognizable relations, which can be
defined by a tuple of automata, one for each component of the
relation, that independently read their components and only
synchronize on their terminal states, \textit{i.e.}, the tuple of states at the
end determines whether the input tuple is accepted. Equivalently, one
can define recognizable relations as finite unions of products of
regular languages. Recognizable relations play a role.
For example,
\cite{BozzelliMP15} use relations over words for
identifying equivalent plays in incomplete information games. The
task is to compute a winning strategy that does not distinguish
between equivalent plays. While this problem is undecidable for
automatic relations, it is possible to synthesize strategies for
recognizable equivalence relations. In view of such results, it
is an interesting question whether one can decide for a given relation
whether it is recognizable.

All these four concepts of automaton definable relations can directly
be adapted to infinite words using the notion of $\omega$-automata
(see \cite{Tho90} for background on $\omega$-automata), leading to the
classes of (deterministic) $\omega$-rational, $\omega$-automatic, and
$\omega$-recognizable relations. Applications like automatic
structures and regular model checking have been adapted to relations
over infinite words, \textit{e.g.}\ by \cite{BlumensathG00,BoigelotLW04}, for instance for
modeling systems with continuous parameters represented by real
numbers (which can be encoded as infinite words, see \textit{e.g.}\ \cite{BoigelotJW05}).

Our contributions are the following, where some background on the
individual results is given below. We note that~(\ref{contr:vpa}) is
not a result on relations over words. It is used in the proof of
(\ref{contr:rec-aut}) but we state it explicitly because we believe
that it is an interesting result on its own.
\begin{enumerate}[(1)]
\item We show that the equivalence problem for binary deterministic
  $\omega$-rational
  relations is undecidable, already for the B\"uchi acceptance
  condition (which is weaker than parity or Muller acceptance
  conditions in the case of deterministic automata).  \label{contr:omega-drat}
\item We show that it is decidable in doubly exponential time for an
  $\omega$-automatic
  relation whether it is
  $\omega$-recognizable.  \label{contr:omega-rec-aut}
\item We reconsider the complexity of deciding for a binary automatic
  relation whether it is recognizable, and prove that it can be done
  in exponential time. \label{contr:rec-aut}
\item\label{contr:vpa} We prove that the regularity problem for deterministic visibly
  pushdown automata --- a model introduced by \cite{AM04} --- is decidable in polynomial time.
\end{enumerate}

The algorithmic theory of deterministic $\omega$-rational
relations has not yet been studied in detail. We think, however, that this class is worth studying in order to understand whether it can
be used in applications that are studied for $\omega$-automatic
relations. One such scenario could be the synthesis of finite state
machines from (binary) $\omega$-automatic
relations. In this setting, an $\omega$-automatic
relation is viewed as a specification that relates input streams to
possible output streams. The task is to automatically synthesize a
synchronous sequential transducer (producing one output letter for each
input letter) that outputs a string for each possible input such that
the resulting pair is in the relation (for instance, \cite{Thomas09} provides an overview of
this kind of automata theoretic synthesis). It has recently
been shown by \cite{FiliotJLW16} that this synchronous synthesis problem can be lifted to
the case of asynchronous automata if the relation is deterministic
rational. This shows that the class of
deterministic rational relations has some interesting properties, and
motivates our study of the corresponding class over infinite words.
Our contribution (\ref{contr:omega-drat}) contrasts the decidability of
equivalence for deterministic rational relations over finite words shown by
\cite{Bir73,HarjuK91} and thus exhibits a difference between
deterministic rational relations over finite and over infinite words.
We prove the undecidability by a reduction from the
intersection emptiness problem for deterministic rational relations
over finite words. The reduction is inspired by a recent construction of \cite{Boeh+} for
proving the undecidability of equivalence for deterministic B\"uchi
one-counter automata.

Contributions~(\ref{contr:omega-rec-aut}) and (\ref{contr:rec-aut})
are about the effectiveness of the hierarchies formed by the four
classes of \maybeomega{}rational,
deterministic \maybeomega{}rational,
\maybeomega{}automatic,
and \maybeomega{}recognizable
relations. A systematic overview and study on the effectiveness of
this hierarchy for finite words is provided by \cite{CCG06}: For a
given rational relation it is undecidable whether it belongs to one of
the other classes, for deterministic rational and automatic relations
it is decidable whether they are recognizable, and the problem of
deciding for a deterministic rational relation whether it is automatic
is open. 
An illustration of those results and our contributions~\ref{contr:omega-rec-aut} and~\ref{contr:rec-aut} can be found in Figure~\ref{figure_overview}.
\begin{figure}
	\centering
	\resizebox{\textwidth}{!}{
		\begin{tikzpicture}[->,>=stealth',shorten >=1pt,auto,node distance=.75cm,
		semithick,align=center]
		\node (rec) {recognizable\\ relations};
		\node [right=of rec] (recinc) {$\subsetneq$};
		\node [right=of recinc] (aut) {automatic\\ relations};
		\node [right=of aut] (autinc) {$\subsetneq$}; 
		\node [right=of autinc] (det) {deterministic rational\\ relations};
		\node [right=of det] (detinc) {$\subsetneq$};
		\node [right=of detinc] (rat) {rational\\ relations};
		
		\path[->]
		(det.south) edge[bend left=10, below, dotted] node{open} (aut.south)
		(aut.south) edge[bend left=10, below] node{\textsc{2ExpTime} (\cite{CCG06})\\ \textsc{ExpTime} for binary relations (Corollary~\ref{corollary_rec_in_aut_decidable_visibly})} (rec.south)
		(det.north) edge[bend right=7, above] node{decidable~(\cite{CCG06}), \textsc{2ExpTime} for binary relations~(\cite{Val75})} (rec.north)
		;
		
		\node [below=7em of recinc] (orecinc) {$\subsetneq$};
		\node [left=of orecinc,xshift=.5em] (orec) {$\omega$-recognizable\\ relations};
		\node [below=7em of autinc] (oautinc) {$\subsetneq$}; 
		\node [left=of oautinc,xshift=.5em] (oaut) {$\omega$-automatic\\ relations};
		\node [below=7em of detinc] (odetinc) {$\subsetneq$};
		\node [left=of odetinc,xshift=.5em] (odet) {$\omega$-deterministic rational\\ relations};
		\node [right=of odetinc,xshift=-.5em] (orat) {$\omega$-rational\\ relations};
		
		\path[->]
		(odet.south) edge[bend left=10, below, dotted] node{open} (oaut.south)
		(oaut.south) edge[bend left=10, below] node{\textsc{2ExpTime} (Theorem~\ref{theorem_rec_in_aut_omega_decidable})} (orec.south)
		(odet.north) edge[bend right=7, above, dotted] node{open} (orec.north)
		;
		\end{tikzpicture}
	}
	\vspace{-.75cm}	\caption{Illustration of decision status for the classes considered in this paper.
		An edge from class $D$ to $C$ either indicates that it is open if it is decidable, given a relation $R\in D$, whether $R\in C$ holds, or states the best known upper bound for this decision problem.
		Deciding whether a rational or $\omega$-rational relation is in one of the subclasses is undecidable (\cite{FR68},\cite{LLP79})}
	\label{figure_overview}
\end{figure}
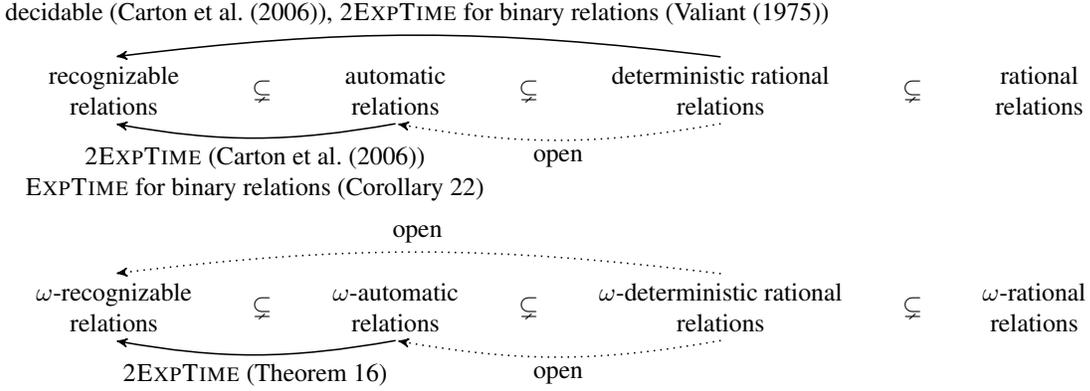

The question of the effectiveness of the hierarchy for relations over
infinite words has already been posed by \cite{Thomas92} (where the
$\omega$-automatic
relations are called B\"uchi recognizable $\omega$-relations).
The undecidability results easily carry over from finite to infinite
words. Our result~(\ref{contr:omega-rec-aut}) lifts one of the two
known decidability results for finite words  to infinite
words. The algorithm is based on a reduction to a problem over finite
words: Using a representation of $\omega$-languages
by finite encodings of ultimately periodic words as demonstrated by \cite{CNP93},
we are able to reformulate the recognizability of an
$\omega$-automatic relation in terms of \emph{slenderness} of a finite number
of languages of finite words.
We adopt the term \emph{slender} from \cite{PS95}.
A language of finite words is called
slender  if there is a bound $k$ such that the language
contains for each length at most $k$ words of this length.
By definition a language is slender if and only if it has \emph{polynomial growth} of order $0$, which is decidable for context-free languages in polynomial time due to \cite{Gaw08}.
We tighten the complexity bound for the slenderness problem for 
nondeterministic finite automata by proving that it is \textsc{NL}-complete.

As mentioned above, the decidability of recognizability of an
automatic relation has already been proved by
\cite{CCG06}. However, the exponential time complexity claimed in that
paper does not follow from the proof presented there.
We illustrate this by a family $R_n$ of automatic relations
  (see Example~\ref{example_rec_in_aut_lower_bound_construction_E}), for
  which an intermediate automaton
  is exponentially larger than the automaton for $R_n$, and the
  procedure from  \cite{CCG06} runs another exponential algorithm on
  this intermediate automaton.
So we revisit the problem and prove the exponential time upper bound for binary
relations based on the 
connection between binary rational relations and pushdown automata: 
For a relation $R$ over finite words, consider the language $L_R$ consisting of the words $\rev(u)\#v$ for all
$(u,v) \in R$, where  $\rev(u)$ denotes the
reverse of $u$.  It turns out that $L_R$ is linear context-free iff $R$ is
rational, $L_R$ is deterministic context-free iff $R$ is
deterministic rational, and $L_R$ is regular iff $R$ is recognizable (\textit{cf.}\ \cite{CCG06}). 
Since $L_R$ is regular iff $R$ is recognizable,  the recognizability test for binary deterministic rational relations 
reduces to
the regularity test for deterministic pushdown automata, which has been shown to be decidable by \cite{Ste67}.
\cite{Val75} improved this decidability result for deterministic pushdown automata by proving a doubly exponential upper bound.\footnote{Recognizability is decidable
for deterministic rational relations of arbitrary arity as shown by \cite{CCG06} but we
are not aware of a proof preserving the doubly exponential runtime.}
We adapt this technique to automatic relations $R$ and show that $L_R$
can in this case be defined by a visibly pushdown automaton (\textsc{VPA})
(see \cite{AM04}), in which the stack operation (pop, push, skip) is
determined by the input symbol, and no $\varepsilon$-transitions are
allowed. The deterministic \textsc{VPA} for $L_R$ is exponential in the size of the
automaton for $R$, and we prove that the regularity test can be done in polynomial
time, our contribution (\ref{contr:vpa}). We note that the polynomial
time regularity test for visibly pushdown processes as presented by
\cite{Srb06} does not imply our result. The model used by \cite{Srb06}
cannot use transitions that cause a pop operation when the stack is
empty.  For our translation from
automatic relations to \textsc{VPA}s we need
these kind of pop operations,
which makes the model different and the decision procedure more
involved (and a reduction to the model of \cite{Srb06} by using new internal symbols to simulate pop operations on the empty stack will not preserve regularity of the language, in general).

This paper is the full version of the conference paper of \cite{LS17}.
In contrast to the conference version, which only sketches proof ideas for the major results, it contains full proofs for all (intermediate) results.
Furthermore, Section \ref{sec:omega-rec-in-aut} is enriched by an example which illustrates why the exponential time complexity claimed by \cite{CCG06} does not follow from their proof approach.
Also, we give an additional comment on the connection of the properties slenderness -- which we exhibit to show the decidability of recognizability for $\omega$-automatic relations -- and finiteness -- used in the approach of \cite{CCG06} to show the decidability of recognizability of automatic relations.

The paper is structured as follows. In Section~\ref{sec:prelim} we
give the definitions of transducers, relations, and visibly pushdown
automata. In Section~\ref{sec:equiv-drat} we prove the undecidability
of the equivalence problem for deterministic $\omega$-rational
relations. Section~\ref{sec:omega-rec-in-aut} contains the decision
procedure for recognizability of $\omega$-automatic relations, and
Section~\ref{sec:rec-in-aut} presents the polynomial time regularity
test for deterministic \textsc{VPA}s and its use for the
recognizability test of automatic relations. Finally, we conclude in
Section~\ref{sec:conclusion}.
 \section{Preliminaries} \label{sec:prelim}
We start by briefly introducing transducers and visibly pushdown automata as well as related terminology as we need them for our results. For more details we refer to \cite{Sak09,FS93,Tho90} and \cite{AM04,BLS06}, respectively.

We denote alphabets (\textit{i.e.}\ finite non-empty sets) by $\Sigma$ and $\Gamma$.
$\Sigma^*$ and $\Sigma^\omega$ are the sets of all finite and infinite words over $\Sigma$, respectively.
Furthermore, a $k$-ary ($\omega$-)relation is a subset $R\subseteq \aryrelation{k}$ or $R\subseteq \aryomegarelation{k}$, respectively.
A unary, \textit{i.e.}\ $1$-ary, ($\omega$-)relation is called a ($\omega$-)language.
Usually we denote languages by $L,K$, \textit{etc.}\ and relations of higher or arbitrary arity by $R,S$, \textit{etc.}
The domain of a relation $R$ is the language 
\[\dom(R):= \{ w\in\Sigma_1^*\mid \exists (v_2,\ldots,v_k)\in \Sigma_2^*\times\ldots\times\Sigma_k^*: (w,v_2,\ldots,v_k)\in R \}.\]
For an $\omega$-relation $R$ the domain is defined analogously.

Lastly, for a natural number $n\in\mathbb{N}$ we define $\underline{n} := \{m \mid 1\leq m\leq n\}$.

\subsection{Finite Transducers and (\omegaintitle{}-)Rational Relations}
A transducer $\mathcal{A}$ is a tuple $(Q,\Sigma_1,\ldots,\Sigma_k,q_0,\Delta,F)$ where $Q$ is the state set, $\Sigma_i$, $1\leq i\leq k$ are (finite) alphabets, $q_0\in Q$ is the initial state, $F\subseteq Q$ denotes the accepting states, and $\Delta\subseteq Q\times (\Sigma_1\cup\{\varepsilon\})\times\ldots\times (\Sigma_k\cup\{\varepsilon\})\times Q$ is the transition relation.
$\mathcal{A}$ is deterministic if there is a state partition $Q=Q_1\cup\ldots\cup Q_k$ such that $\Delta$ can be interpreted as partial function $\delta : \bigcup_{j=1}^k(Q_j\times (\Sigma_j\cup\{\varepsilon\}))\rightarrow Q$ with the restriction that if $\delta(q,\varepsilon)$ is defined then no $\delta(q,a)$, $a\neq \varepsilon$ is defined.
Note that the state determines which component the transducer processes.
$\mathcal{A}$ is \emph{complete} if $\delta$ is total (up to the restriction for $\varepsilon$-transitions).

A run of $\mathcal{A}$ on a tuple $u\in \aryrelation{k}$ is a sequence $\rho = p_0\ldots p_n\in Q^*$ such that there is a decomposition $u = (a_{1,1},\ldots,a_{1,k}) \ldots (a_{n,1},\ldots,a_{n,k})$ where the $a_{i,j}$ are in $\Sigma_j\cup\{\varepsilon\}$ and for all $i\in\{1,\ldots,n\}$ it holds that $(p_{i-1},a_{i,1},\ldots,a_{i,k},p_i)\in\Delta$.
The run of $\mathcal{A}$ on a tuple over infinite words in $\aryomegarelation{k}$ is an infinite sequence $p_0p_1\ldots\in Q^\omega$ defined analogously to the case of finite words.
We use the shorthand notation $\mathcal{A}: p\xrightarrow{u_1\smash{/}\ldots\smash{/}u_k}q$ or $\mathcal{A}: p\xrightarrow{u}q$ to denote the existence of a run of $\mathcal{A}$ on $u = (u_1,\ldots,u_k)$, $u_j\in\Sigma_j^*$ starting in $p$ and ending in $q$.
Moreover, $\mathcal{A}: p\smash{\xrightarrow[\smash{F}]{u}}q$ denotes the existence of a run from $p$ to $q$ which contains an accepting state.
Concerning runs on tuples of infinite words we deliberately extend this notation in the natural way and write  $p\xrightarrow{u}(q\xrightarrow{v_i}q)_{i\geq 0}$ or $p\xrightarrow{u}(q\xrightarrow{v}q)^\omega$ if both, the run and the input tuple, permit it.
A run on $u$ is called \emph{accepting} if it starts in the initial state $q_0$ and ends in an accepting state.
Moreover, $\mathcal{A}$ accepts $u$ if there is an accepting run starting of $\mathcal{A}$ on $u$.
Then $\mathcal{A}$ defines the relation $R_*(\mathcal{A})\subseteq \aryrelation{k}$ containing precisely those tuples accepted by $\mathcal{A}$.
To enhance the expressive power of deterministic transducers, the relation $R_*(\mathcal{A})$ is defined as the relation of all $u$ such that $\mathcal{A}$ accepts $u(\#,\ldots,\#)$ for some fresh fixed symbol $\#\notin \bigcup_{j=1}^k\Sigma_j$.
The relations definable by a (deterministic) transducer are called \emph{(deterministic) rational relations}.
For tuples over infinite words $u\in\aryomegarelation{k}$ we utilize the B\"uchi condition (\textit{cf.}\ \cite{Buec62}). That is, a run $\rho\in Q^\omega$ is accepting if it starts in the initial state $q_0$ and a state $f\in F$ occurs infinitely often in $\rho$.
Then $\mathcal{A}$ accepts $u$ if there is an accepting run of $\mathcal{A}$ on $u$ and $R_\omega(\mathcal{A}) \subseteq\aryomegarelation{k}$ is the relation of all tuples of infinite words accepted by $\mathcal{A}$.
We refer to $\mathcal{A}$ as B\"uchi transducer if we are interested in the relation of infinite words defined by it.
The class of $\omega$-rational relations consists of all relations definable by B\"uchi transducers.

It is well-known that  deterministic B\"uchi automata are not sufficient to capture the $\omega$-regular languages (see \cite{Tho90}) which are the $\omega$-rational relations of arity one.
Therefore, we use another kind of transducer to define deterministic $\omega$-rational relations: a deterministic parity transducer is a tuple $\mathcal{A} = (Q,\Sigma_1,\ldots,\Sigma_k,q_0,\delta,\Omega)$ where the first $k+3$ items are the same as for deterministic transducers and $\Omega: Q\rightarrow \mathbb{N}$ is the priority function.
A run is accepting if it starts in the initial state and the maximal priority occurring infinitely often in the run is even (\textit{cf.}\ \cite{Pit06}).

A transducer is synchronous if for each pair $(p,a_1,\ldots,a_k,q),(q,b_1,\ldots,b_k,r)$ of successive transitions it holds that $a_j=\varepsilon$ implies $b_j=\varepsilon$ for all $j\in\{1,\ldots,k\}$.
Intuitively, a synchronous transducer is a finite automaton over the vector alphabet $\Sigma_1\times\ldots\times\Sigma_k$ and, if it operates on tuples $(u_1,\ldots,u_k)$ of finite words, the components $u_j$ may be of different length (\textit{i.e.}\ if a $u_j$ has been  processed completely, the transducer may use transitions reading $\varepsilon$ in the $j$-th component to process the remaining input in the other components).
In fact, synchronous transducers inherit the rich properties of finite automata -- \textit{e.g.}, they are closed under all Boolean operations and can be determinized.
In particular,  synchronous (nondeterministic) B\"uchi transducer and deterministic synchronous parity transducer can be effectively transformed into each other (see \cite{Sak09,FS93,Pit06}).
Synchronous (B\"uchi) transducers define the class of $(\omega\text{-})$automatic relations.

Finally, the last class of relations we consider are \maybeomega{}recognizable relations. 
A relation $R\subseteq\aryrelation{k}$ (or $R\subseteq\aryomegarelation{k}$) is \emph{\maybeomega{}recognizable} if it is the finite union of direct products of \maybeomega{}regular languages --- \textit{i.e.}\ $R = \bigcup_{i=1}^\ell L_{i,1}\times\ldots\times L_{i,k}$ where the $L_{i,j}$ are ($\omega\text{-}$)regular languages.

It is well-known that the classes of \maybeomega{}recognizable, \maybeomega{}automatic, deterministic \maybeomega{}ra\-tio\-nal relations, and \maybeomega{}rational relations form a strict hierarchy (see \cite{Sak09}).

\subsection{Visibly Pushdown Automata}
In Section~\ref{sec:rec-in-aut}, we use visibly pushdown automata (\textsc{VPA}s) which have been introduced by \cite{AM04}. They operate on typed alphabets, called pushdown alphabets below, where the type of input symbol determines the stack operation.
Formally, a \emph{pushdown alphabet} is an alphabet $\Sigma$ consisting of three disjoint parts --- namely, a set $\Sigma_c$ of \emph{call symbols} enforcing a push operation, a set $\Sigma_r$ of \emph{return symbols} enforcing a pop operation and internal symbols $\Sigma_\text{int}$ which do not permit any stack operation.
A \textsc{VPA} is a tuple $\mathcal{P} = (P,\Sigma,\Gamma,p_0,\bot,\Delta,F)$ where $P$ is a finite set of states, $\Sigma = \Sigma_c~\dot{\cup}~\Sigma_r~\dot{\cup}~\Sigma_\text{int}$ is a finite pushdown alphabet, $\Gamma$ is the stack alphabet and $\bot\in\Gamma$ is the stack bottom symbol, $p_0\in P$ is the initial state, $\Delta\subseteq (P\times\Sigma_c\times P\times(\Gamma\setminus\{\bot\}))\cup(P\times\Sigma_r\times\Gamma\times P)\cup (P\times \Sigma_\text{int}\times P)$ is the transition relation, and $F$ is the set of accepting states.

A \emph{configuration} of $\mathcal{P}$ is a pair in $(p,\alpha)\in P\times(\Gamma\setminus\{\bot\})^*\{\bot\}$ where $p$ is the current state of $\mathcal{P}$ and $\alpha$ is the current stack content ($\alpha[0]$ is the top of the stack).
Note that the stack bottom symbol $\bot$ occurs precisely at the bottom of the stack.
The stack whose only content is $\bot$, is called the \emph{empty stack}.
$\mathcal{P}$ can \emph{proceed} from a configuration $(p,\alpha)$ to another configuration $(q,\beta)$ \emph{via} $a\in\Sigma$ if
	$a\in\Sigma_c$ and there is a $(p,a,q,\gamma)\in \Delta\cap (P\times\Sigma_c\times P\times(\Gamma\setminus\{\bot\}))$ such that $\beta = \gamma\alpha$ (\emph{push} operation),
	$a\in\Sigma_r$ and there is a $(p,a,\gamma,q)\in \Delta\cap (P\times\Sigma_r\times\Gamma\times P)$ such that $\alpha = \gamma\beta$ or $\gamma = \alpha = \beta = \bot$ --- that is, the empty stack may be popped arbitrarily often (\emph{pop} operation),
	or $a\in\Sigma_\text{int}$ and there is a $(p,a,q)\in \Delta\cap (P\times\Sigma_\text{int}\times P)$ such that $\alpha = \beta$ (\emph{noop}).
A \emph{run} of $\mathcal{P}$ on a word $u = a_1,\ldots,a_n\in\Sigma^*$ is a sequence of configurations $(p_1,\alpha_1)\ldots (p_{n+1},\alpha_{n+1})$ connected by transitions using the corresponding input letter.
As for transducers we use the shorthand $\mathcal{P}:(p_1,\alpha_1)\xrightarrow{u}(p_{n+1},\alpha_{n+1})$ to denote a run.
A run is accepting if it starts with the initial configuration $(p_0,\bot)$ and ends in a configuration $(p_f,\alpha_f)$ with $p_f \in F$. We say that $\mathcal{P}$ accepts $u$ if there is an accepting run of $\mathcal{P}$ on $u$ and write $L(\mathcal{P})$ for the language of all words accepted by $\mathcal{P}$.
Furthermore, $\mathcal{P}$ accepts $u$ from $(p,\alpha)$ if there is a run of $\mathcal{P}$ starting in $(p,\alpha)$ and ending in a configuration $(p_f,\alpha_f)$ with $p_f\in F$.
We write $L(p,\alpha)$ for the set of all words accepted from the configuration $(p,\alpha)$.
Note that for the initial configuration $(p_0,\bot)$ we have that $L(\mathcal{P}) = L(p_0,\bot)$.
Two configurations $(p,\alpha),(q,\beta)$ are $\mathcal{P}$-equivalent if $L(p,\alpha) = L(q,\beta)$.
We denote the $\mathcal{P}$-equivalence relation by $\approx_\mathcal{P}$ .
That is, $(p,\alpha)\approx_\mathcal{P}(q,\beta)$ if and only if $L(p,\alpha) = L(q,\beta)$.
Lastly, a configuration $(p,\alpha)$ is \emph{reachable} if there is a run from $(p_0,\bot)$  to $(p,\alpha)$.

A \emph{deterministic} \textsc{VPA} (\textsc{DVPA})  $\mathcal{P}$ is a \textsc{VPA} that can proceed to at most one configuration for each given configuration and $a\in \Sigma$.

Viewing the call symbols as opening and the return symbols as closing parenthesis, one obtains a natural notion of a return matching a call, and unmatched call or return symbols. 
Furthermore, we need the notion of well-matched words (\textit{cf.}\ \cite{BLS06}).
The set of well-matched words over a pushdown alphabet $\Sigma$ is defined inductively by the following rules:
\begin{itemize}
	\item Each $w\in \Sigma_\text{int}^*$ is a well-matched word.
	\item For each well-matched word $w$, $c\in\Sigma_c$ and $r\in\Sigma_r$ the word $cwr$ is well-matched.
	\item Given two well-matched words $w,v$ their concatenation $wv$ is well-matched.
\end{itemize}
An important observation regarding well-matched words is that the behavior of $\mathcal{P}$ on a well-matched word $w$ is invariant under the stack content.
That is, for any configurations $(p,\alpha), (p,\beta)$ we have that $\mathcal{P}: (p,\alpha)\xrightarrow{w}(q,\alpha)$ if and only if $\mathcal{P}:(p,\beta)\xrightarrow{w}(q,\beta)$.
In particular, this holds true for the empty stack $\alpha =\bot$.
Furthermore, every word $u\in\Sigma^*$ with $\mathcal{P}: (q_0,\alpha)\xrightarrow{u}(p_{n+1},\gamma_n\ldots\gamma_1\alpha)$ for some $\gamma_n,\ldots,\gamma_1\in\Gamma\setminus\{\bot\}$, $n\leq |u|$ can be uniquely decomposed into a prefix $u'\in\Sigma^*$, well-matched words $w_1,\ldots, w_{n+1}$, and call symbols $c_1,\ldots,c_n\in\Sigma_c$ such that $u=u'w_1c_1w_2\ldots c_nw_{n+1}$ and $u'$ is minimal (in other words, $u'$ is the shortest prefix that contains unmatched return symbols).
Moreover, there are configurations \[(p_1,\alpha),(p_1',\gamma_1\alpha),\ldots,(p_n,\gamma_{n-1}\ldots\gamma_{1}\alpha),(p_n',\gamma_{n}\ldots\gamma_{1}\alpha)\]
such that 
\[(p_{i-1}',\gamma_{i-1}\ldots\gamma_1\alpha)\xrightarrow{w_i}(p_i,\gamma_{i-1}\ldots\gamma_1\alpha)\xrightarrow{c_i}(p_i,\gamma_{i}\ldots\gamma_{1}\alpha)\]
holds.
That is, informally, the symbol $c_i$ is responsible for pushing $\gamma_i$ onto the stack.

A similar unique decomposition is possible for words $u$ popping a sequence $\gamma_n\ldots\gamma_1$ from the top of the stack.
In that case we have that the word $u$ factorizes into $u=w_nr_n\ldots w_1r_1w_0u'$ where the $w_j$ are well-matched words, the $r_i$ are return symbols responsible for popping the $\gamma_i$, and $u'$ is a minimal suffix.
 \section{The Equivalence Problem for Deterministic B\"uchi Transducers} \label{sec:equiv-drat}
In this section we show that the equivalence for deterministic B\"uchi transducers is undecidable -- in difference to its analogue for relations over finite words proven by \cite{Bir73,HarjuK91}.
Our proof is derived from a recent construction by \cite{Boeh+} for proving that the equivalence problem for one-counter B\"uchi automata is undecidable.
We reduce the intersection emptiness problem for relations over finite words to the equivalence problem for deterministic B\"uchi transducers.
\begin{proposition}[\cite{RS59,berstel79}]\label{proposition_emptiness_intersection_undecidable}
	The intersection emptiness problem, asking for two binary relations given by deterministic transducers $\mathcal{A},\mathcal{B}$ whether $R_*(\mathcal{A})\cap R_*(\mathcal{B}) = \emptyset$ holds, is undecidable.
\end{proposition}

\begin{theorem}\label{theorem_equivalence_undecidable}
	The equivalence problem for $\omega$-rational relations of arity at least two is undecidable for deterministic B\"uchi transducers.
\end{theorem}
\begin{proof}
	We prove Theorem \ref{theorem_equivalence_undecidable} by providing a many-one-reduction from the emptiness intersection problem over finite relations to the equivalence problem for deterministic $\omega$-rational relations.
	Then the claim follows due to the undecidability of the emptiness intersection problem (\textit{cf.}\ Proposition \ref{proposition_emptiness_intersection_undecidable}).
	Furthermore, it suffices to provide the reduction for relations of arity $k=2$.
	For $k > 2$ the claim follows by adding dummy components to the relation.

	 Let $\mathcal{A}_R$, $\mathcal{A}_S$ be deterministic transducers defining binary relations $R$ and $S$ over finite words, respectively.
	 More precisely, we let
	 \[\mathcal{A}_R = (Q_R,\Sigma_1,\Sigma_2,q_0^R,\delta_R,F_R)\text{ and }\mathcal{A}_S = (Q_S,\Sigma_1,\Sigma_2,q_0^S,\delta_S,F_S).\]
	 We construct deterministic B\"uchi transducers $\mathcal{B}_R$ and $\mathcal{B}_S$ such that
	 \[R\cap S\neq\emptyset \Leftrightarrow R_\omega(\mathcal{B}_R)\neq R_\omega(\mathcal{B}_S). \]
	 That is, each tuple in $R\cap S$ induces a witness for $R_\omega(\mathcal{B}_R) \neq R_\omega(\mathcal{B}_S)$ and vice versa.
	
	 Recall that $\mathcal{A}_R$, $\mathcal{A}_S$ accept a tuple $(u,v)$ if there is an accepting run on $(u,v)(\#,\#)$ (where $\#$ is an endmarker symbol not contained in any alphabet involved).
	 Then it is easy to see that we can assume that the deterministic transducers $\mathcal{A}_R$ and $\mathcal{A}_S$ are in normal form according to \cite{Sak09}: the initial states $q_0^R$ and $q_0^S$ do not have incoming transitions and there are unique accepting states $q_a^R$ and $q_a^S$ as well as rejecting states $q_r^R$ and $q_r^S$ that
	 \begin{enumerate}
	 	\item are entered only by transitions labeled $\#$, and
	 	\item have no outgoing transitions.
	 \end{enumerate}
	 That is, we have that $F_R = \{q_a^R\}$ and upon the end of any run $\mathcal{A}_R$ is either in state $q_a^R$ or $q_r^R \neq q_a^R$ after reading the endmarker $\#$ in both components.
	 Analogously, the same applies for $\mathcal{A}_S$.

	The construction of $\mathcal{B}_R$ and $\mathcal{B}_S$ is illustrated in Figure \ref{figure_equivalence_buechi_undecidable_construction}.
	\begin{figure}
		\centering
		\resizebox{\textwidth}{!}{
			\begin{tikzpicture}[->,>=stealth',shorten >=1pt,auto,node distance=2.6cm,
			semithick,align=center]
			\tikzstyle{automaton-box}=[draw=black, fill=none, thick,
			rectangle, rounded corners, minimum width=5cm, minimum height=3cm,text width=5cm,node distance=10cm, align=left, text depth=3cm]
			
			\node[automaton-box] (AR) [] {$\mathcal{A}_R$};
			\node[state, initial, initial text={$\mathcal{B}_R$}] (q0R) [right=0.3cm of AR.west] {$q_0^R$};
			\node[state, accepting] (qaR) [left=.5cm of AR.east,yshift=.75cm] {$q_a^R$};
			\node[state, accepting] (qrR) [left=.5cm of AR.east,yshift=-.75cm] {$q_r^R$};
			\node[] (preqaR) [left=1cm of qaR] {$\ldots$};
			\node[] (preqrR) [left=1cm of qrR] {$\ldots$};
			\node[] (postq0R) [right=.75cm of q0R] {$\ldots$};
			
			\node[automaton-box, node distance=7cm] (AS) [right of=AR] {$\mathcal{A}_S$};
			\node[state, initial, initial text={$\mathcal{B}_S$}] (q0S) [right=0.3cm of AS.west] {$q_0^S$};
			\node[state, dashed] (qaS) [left=.5cm of AS.east,yshift=.75cm] {$q_a^S$};
			\node[state, accepting] (qrS) [left=.5cm of AS.east,yshift=-.75cm] {$q_r^S$};
			\node[] (preqaS) [left=1cm of qaS] {$\ldots$};
			\node[] (preqrS) [left=1cm of qrS] {$\ldots$};
			\node[] (postq0S) [right=.75cm of q0S] {$\ldots$};
			
			\path
			(preqaR) edge [] node{$\#/\#$} (qaR)
			(preqrR) edge [] node{$\#/\#$} (qrR)
			(preqaS) edge [] node{$\#/\#$} (qaS)
			(preqrS) edge [] node{$\#/\#$} (qrS)
			(q0R) edge [] node{} (postq0R)
			(q0S) edge [] node{} (postq0S)
			
			(qaR) edge [above, pos=.8] node{$\varepsilon/\varepsilon$} (q0R.north east)
			(qrR) edge [bend right] node{$\varepsilon/\varepsilon$} (q0S.south west)
			(qaS) edge [above, pos=.8] node{$\varepsilon/\varepsilon$} (q0S.north east)
			(qrS) edge [bend left=25,above,pos=.45] node{$\varepsilon/\varepsilon$} (q0R.south east)
			;
			\end{tikzpicture}
		}
		\vspace{-1cm}		\caption{Illustration of the transducers $\mathcal{B}_R$, $\mathcal{B}_S$.
			The labels  $\#/\#$ are just used for comprehensibility.
			In the formal construction the $\#$ symbols are read in succession and the transducers may even read other symbols between them (but only in the component where no $\#$ has been read yet).}
		\label{figure_equivalence_buechi_undecidable_construction}
	\end{figure}
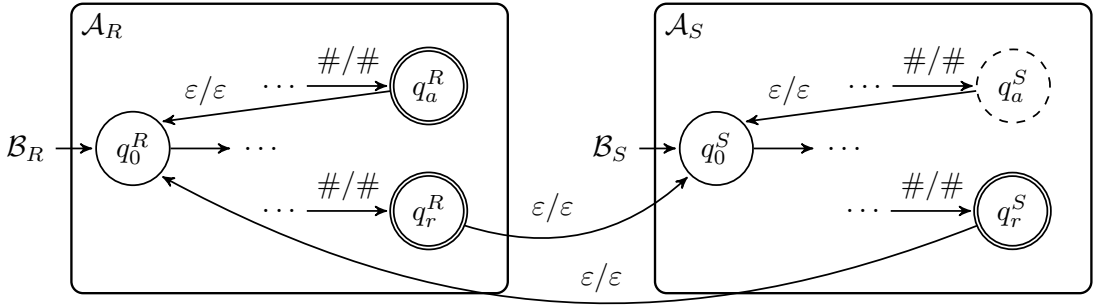
	
	 Both B\"uchi transducers are almost the same except for the initial state: both consist of the union of the transition structures of $\mathcal{A}_R$ and $\mathcal{A}_S$ complemented by transitions labeled $\varepsilon/\varepsilon$ from $q_a^X$ to $q_0^X$ and $q_r^X$ to $q_0^{Y}$ for $X,Y\in\{R,S\}$, $X\neq Y$.
	That is, upon reaching a rejecting state of $\mathcal{A}_R$ or $\mathcal{A}_S$ the new transducers will switch to the initial state of the other subtransducer\footnote{For our purpose, a subtransducer of $\mathcal{B}$ is a transducer obtained from $\mathcal{B}$ by removing states and transitions. Also, accepting states do not have to be preserved.
	In particular, $\mathcal{A}_R$ and $\mathcal{A}_S$ are subtransducers of $\mathcal{B}_R$ and $\mathcal{B}_S$ by definition.} and upon reaching an accepting state they will return to the initial state of the current subtransducer.
	The new accepting states are $q_a^R,q_r^R,q_r^S$ .
	Note that $q_a^S$ is not accepting introducing an asymmetry.
	Finally, the initial state of $\mathcal{B}_X$ is $q_0^X$.
	Formally, we set
	\[\mathcal{B}_R := (Q_\mathcal{B},\Sigma_1',\Sigma_2',q_0^R,\delta_\mathcal{B},F_\mathcal{B})\text{ and }\mathcal{B}_S := (Q_\mathcal{B},\Sigma_1',\Sigma_2',q_0^S,\delta_\mathcal{B},F_\mathcal{B})\] where
	\begin{itemize}
		\item $Q_\mathcal{B} := Q_R ~\dot{\cup}~ Q_S$,
		\item $\Sigma_i' := \Sigma_i \cup \{\#\}$, $i\in\{1,2\}$, and
		\item $F_\mathcal{B} := \{q_a^R, q_r^R, q_r^S\}$.
	\end{itemize}
	The transition relation $\delta_\mathcal{B}$ is defined as follows:
	\[\delta_\mathcal{B}(q,a) := \begin{cases}
	\delta_R(q,a),\quad &q\in Q_R\setminus\{q_a^R,q_r^R\}\\
	\delta_S(q,a),\quad &q\in Q_S\setminus\{q_a^S,q_r^S\}\\
	\delta_S(q_0^S,a),\quad &q\in \{q_r^R, q_a^S\}\\
	\delta_R(q_0^R,a),\quad &q\in \{q_r^S, q_a^R\}\\
	\end{cases} \]

	Further on, we show the correctness of our construction.
	Pick a tuple $(u,v)$ in $R\cap S$.
	We have to show that $R_\omega(\mathcal{B}_R)\neq R_\omega(\mathcal{B}_S)$.
	Then the unique runs of $\mathcal{A}_R$ and $\mathcal{A}_S$ on $(u\#,v\#)$ end in $q_a^R$ and $q_a^S$, respectively:
	\[\mathcal{A}_R:~q_0^R\xrightarrow{u\#/v\#} q_a^R\text{ and } \mathcal{A}_S:~q_0^S\xrightarrow{u\#/v\#} q_a^S.\]
	Recall that $q_a^R$ and $q_a^S$ have precisely the same transitions as $q_0^R$ and $q_0^S$, respectively. Thus, they have the same behavior.
	Hence, the unique runs of $\mathcal{B}_R$ and $\mathcal{B}_S$ on $w := (u\#,v\#)^\omega$, which are completely determined by $\mathcal{A}_R$ and $\mathcal{A}_S$, have the following shape:
	\[\mathcal{B}_R:~q_0^R\xrightarrow{u\#/v\#} \left(q_a^R \xrightarrow{u\#/v\#} q_a^R\right)^\omega \text{ and } \mathcal{B}_S:~q_0^S\xrightarrow{u\#/v\#} \left(q_a^S \xrightarrow{u\#/v\#} q_a^S\right)^\omega.\]
	Since $q_a^R\in F_\mathcal{B}$, it follows that $w\in R_\omega(\mathcal{B}_R)$.
	On the other hand, the run of $\mathcal{B}_S$ stays completely in the $\mathcal{A}_S$ subtransducer and $q_r^S$ does not occur in it. Assuming otherwise, $\mathcal{A}_S$ would reject $(u,v)$ which would be a contradiction.
	But then no state in $F_\mathcal{B}$ occurs in the run of $\mathcal{B}_S$ (recall that in contrast to $q_a^R$ the state $q_a^S$ is not in $F_\mathcal{B}$).
	Hence, $w\notin R_\omega(\mathcal{B}_S)$.
	Therefore, the induced unique run of $\mathcal{B}_R$ on $(u\#,v\#)^\omega$ is accepting while the unique run of $\mathcal{B}_S$ is rejecting.
	Thus, $w\in R_\omega(\mathcal{B}_R)\setminus R_\omega(\mathcal{B}_S)$ and we can conclude that $R_\omega(\mathcal{B}_R)\neq R_\omega(\mathcal{B}_S)$ holds.

	For the other direction, suppose $R_\omega(\mathcal{B}_R)\neq R_\omega(\mathcal{B}_S)$ holds.
	Then there is a pair of infinite words $(u,v)$ that is rejected by one of the transducers,  and accepted by the other.
	Recall that the accepting states of both $\mathcal{B}_R$ and $\mathcal{B}_S$ can only be entered by reading the $\#$-symbol in both components. 
    Hence, both components of $(u,v)$  have to contain infinitely often $\#$, since one of the B\"uchi transducers accepts.
    More precisely, the pair $(u,v)$ can be written as
    \[
    (u,v) = (u_0,v_0)(\#,\#)(u_1,v_1)(\#,\#)\ldots \in R_\omega(\mathcal{B}_R)\triangle R_\omega(\mathcal{B}_S),\text{ with } (u_i,v_i)\in\Sigma^*_1\times\Sigma_2^*~\forall i\in\mathbb{N}.
    \]
    We claim that there is a $p\in\mathbb{N}$ such that $(u_p,v_p)\in R\cap S$.
    Then $R\cap S\neq\emptyset$ follows immediately.
  
  	Let $\rho_R$ and $\rho_S$ be the unique runs of $\mathcal{B}_R$ and $\mathcal{B}_S$, respectively.
  	W.l.o.g.\ assume that $\mathcal{B}_R$ rejects $(u,v)$ while $\mathcal{B}_S$ accepts it.
  	In the other case the reasoning is exactly the same with $\mathcal{B}_R$ and $\mathcal{B}_S$ exchanged.
  
   Since $\rho_R$ is not accepting, the states in $F_\mathcal{B}$ occur only finitely often in $\rho_R$. On the other hand, the endmarker $\#$ is read infinitely often. Thus, states in $F_\mathcal{B}\cup \{q_a^S\}$ occur infinitely often because they are entered if and only if endmarkers have been read in both components.
   It follows that $\rho_R$ stays in the $\mathcal{A}_S$ subtransducer from some point on and $q_a^S$ occurs infinitely often in $\rho_R$. To be more precise, $q_a^S$ occurs precisely after reading the endmarker in both components and, afterwards, the run continues in $\mathcal{A}_S$.
   All in all, the run is determined by run fragments
   \[q_0^R \xrightarrow{u_0\#u_1\#\ldots u_{j-1}\#/v_0\#v_1\#\ldots v_{j-1}\#} q_0^S\xrightarrow{u_{j}\#/v_{j}\#} q_a^S\text{ and }  q_a^S\xrightarrow{u_{i}\#/v_{i}\#} q_a^S\]
   for some $j\in\mathbb{N}$ and all $i> j$.
   Hence, it holds that
   \begin{align}
   \exists j\geq 0~\forall i\geq j: (u_i,v_i)\in S\tag{$\star$}\label{equivalence_buechi_undecidable_proof_star}
   \end{align}
   because $\mathcal{A}_S$ and $\mathcal{B}_R$ are deterministic and, by construction, $q_a^S$ imitates $q_0^S$.
   
   Further on, it suffices to show that $\rho_S$ (the unique accepting run of $\mathcal{B}_S$) stays in $\mathcal{A}_R$ from some point on.
   Then it follows analogously to the case for $\rho_R$ that \eqref{equivalence_buechi_undecidable_proof_star} holds for $R$ and, thus, we have that $(u_p,v_p)\in R\cap S$ for some $p$.
   First of all, $\rho_S$ does not stay in $\mathcal{A}_S$ from some point on. Otherwise, it would not be accepting.
   Assume for the sake of contradiction that the run $\rho_S$ switches infinitely often between the two subtransducers, \textit{i.e.}\ both $q_r^S$ and $q_r^R$ occur infinitely often in the run.
   It follows that $\rho_S$ contains infinitely many fragments of the form
   \[q_r^R\xrightarrow{u_i\#/v_i\#} q_r^S\text{ or }q_a^S\xrightarrow{u_i\#/v_i\#} q_r^S \]
   where all intermediate states are in $Q_S$.
   Thus, because $q_r^R$ and $q_a^S$ behave in the same way as $q_0^S$ by construction and $\mathcal{A}_S$ is deterministic, we have that
   \[ \forall j\geq 0~\exists i\geq j: (u_i,v_i)\notin S. \]
   But this is a direct contradiction to \eqref{equivalence_buechi_undecidable_proof_star}.
   Hence, $\rho_S$ stays in the $\mathcal{A}_R$ subtransducer from some point on, and similarly to the case for $\rho_R$ above, it follows that
   \begin{align}
   \exists j\geq 0~\forall i\geq j: (u_i,v_i)\in R.\tag{$\star\star$}\label{equivalence_buechi_undecidable_proof_star2}
   \end{align}
   Let $p\in\mathbb{N}$ the maximum of the existentially quantified $j$'s in \eqref{equivalence_buechi_undecidable_proof_star} and \eqref{equivalence_buechi_undecidable_proof_star2}.
   Then we have that $(u_p,v_p)\in R\cap S\neq \emptyset$.

   All in all, we have shown that $R\cap S\neq\emptyset \Leftrightarrow R_\omega(\mathcal{B}_R)\neq R_\omega(\mathcal{B}_S)$ holds and, thus, the correctness of our reduction.
\end{proof}

We note that our reduction is rather generic an could be applied to other classes of automata for which the intersection emptiness problem on finite words is undecidable.
 \section{Deciding Recognizability of \texorpdfstring{$\boldsymbol{\omega}$}{omega}-Automatic Relations} \label{sec:omega-rec-in-aut}
Our aim in this section is to decide $\omega$-recognizability of $\omega$-automatic relations in doubly exponential time.
That is, given a deterministic synchronous transducer, decide whether it defines an $\omega$-recognizable relation.
The proof approach is based on an algorithm for relations over finite words given by \cite{CCG06} which we briefly discuss in Subsection \ref{subsec:omega-rec-in-aut-revision}.
Afterwards, in Subsection \ref{subsec:omega-rec-in-aut-infinite}, we present our main result of this section.
In Subsection \ref{subsec:omega-rec-in-aut-slenderness-vs-finiteness}, we comment on a connection between \cite{CCG06}'s original proof and our (alternative) approach for infinite words.

\subsection{Revision: Deciding Recognizability of Automatic Relations}\label{subsec:omega-rec-in-aut-revision}
Let $R$ be an \maybeomega{}automatic relation of arity $k$.
For each $j\leq k$ we define the equivalence relation
\begin{multline*}
	E_j := \{((u_1,\ldots,u_j),(v_1,\ldots,v_j))\mid \forall w_{j+1},\ldots,w_k:\\
	(u_1,\ldots,u_j,w_{j+1},\ldots,w_k)\in R \Leftrightarrow  (v_1,\ldots,v_j,w_{j+1},\ldots,w_k)\in R \}.
\end{multline*}
Then the key to decide \maybeomega{}recognizability is the following result which has been proven by \cite{CCG06} for relations over finite words and is easily extensible to infinite words:
\begin{lemma}[\cite{CCG06}]
	\label{lemma_restrictions_recognizable_finite_implies_rec_omega_automatic}
	Let $R$ be an \maybeomega{}automatic relation of arity $k$.
	Then for all $1\leq j\leq k$ the equivalence relation $E_j$ has finite index if and only if $R$ is \maybeomega{}recognizable.
\end{lemma}
Here we shall adapt the proof of \cite{CCG06} to $\omega$-automatic relations.

\begin{proof}
	\underline{$\Rightarrow$}: Assume $R\neq \emptyset$ is $\omega$-recognizable. Then we have that $R = \bigcup_{i=1}^\ell L_{i,1}\times\ldots\times L_{i,k}$ for some $\omega$-regular languages $L_{i,n}\subseteq \Sigma^\omega_n$.
	Pick words $u_i\in L_{i,1}$ for all $i\in\{1,\ldots,\ell \}$.
	Then each word in $\Sigma_1^\omega$ is either equivalent w.r.t.\ $E_1$ to one of these $u_i$ or belongs to the equivalence class $\Sigma_1^\omega\setminus\dom(R)$.
	Thus, $E_1$ has finite index.
	The proof for $E_j$, $j>1$ works similarly by picking tuples $(u_{i,1},\ldots,u_{i,j})\in L_{i,1}\times\ldots\times L_{i,j}$.
	
	\underline{$\Leftarrow$}: Assume all $E_j$ have finite index.
	We show the claim by induction over $k$.
	For the base case of a $1$-ary $\omega$-rational relation the claim is trivial.
	Suppose that $k\geq 2$.
	We prove that $R$ can be written as finite union of direct products of $\omega$-regular languages.
	Then it follows by definition that $R$ is recognizable.
	Recall that $\dom(R)$ is $\omega$-regular.
	Hence, we can pick an ultimately periodic word $u_1\in \dom(R)$ (\textit{cf.}\ \cite{Buec62}).
	Clearly, $\{u_1\}\times \Sigma^\omega_1$ is $\omega$-automatic.
	Furthermore, since $\omega$-automatic relations are closed under intersection (and $E_1$ itself is $\omega$-automatic) it holds that
	\[E_1\cap(\{u_1\}\times\Sigma_1^\omega) = \{u_1\}\times \{w\in \Sigma_1^\omega \mid (u_1,w)\in E_1 \} = \{u_1\}\times [u_1]_{E_1} \]
	is $\omega$-automatic.
	In particular, it follows that the equivalence class of $u_1$, denoted $[u_1]_{E_1}$, is $\omega$-regular.
	Moreover, each $\{u_1 \}\times \Sigma^\omega_2\times\ldots\times\Sigma^\omega_k$ is $\omega$-recognizable.
	Hence, the relation
	\[R|_{u_1} := \{(x_2,\ldots,x_k)\mid (u_1,x_2,\ldots,x_k)\in R \}\subseteq\Sigma^\omega_2\times\ldots\times\Sigma^\omega_k\]
	is $\omega$-automatic.
	By iteratively applying this reasoning to $\dom(R)\setminus [u_1]_{E_1}$ we obtain a sequence $u_1,u_2,\ldots$ such that all $u_i$ are pairwise non-equivalent ultimately periodic representatives.
	Since $E_1$ has finite index by assumption, this sequence is finite.
	Observe that this implies that $R$ can be written as the finite union $R=\bigcup_{i=1}^\ell [u_i]_{E_1}\times R|_{u_i}$.
	Recall that the $R|_{u_i}$ are $\omega$-automatic.
	Hence, it suffices to show that for each $R|_{u_1}$ the induced equivalence relations have all finite index.
	Then the claim follows by the induction hypothesis.
	For this purpose, observe that an equivalence class of some $E_j$ is solely determined by a set of possible outputs.
	More precisely, for $T:=R|_{u_1}$ and all $2\leq j\leq k$ we have that
	\begin{multline*}
	T|_{v_2,\ldots,v_j} = \{(v_{j+1},\ldots,v_k)\mid (v_2,\ldots,v_k)\in T\}\\
	= \{(v_{j+1},\ldots,v_k)\mid (u_i,v_2,\ldots,v_k)\in R\} = R|_{u_i,v_2,\ldots,v_j}.
	\end{multline*}
	Therefore, it follows by the definition of the $E_j$ that for all $R|_{u_i}$ the induced equivalence relations have all finite index.
\end{proof}

Based on that lemma, the recognizability test presented by \cite{CCG06} proceeds as follows. It is shown that each $E_j$ is an automatic equivalence relation by constructing a synchronous transducer for $E_j$. It remains to decide for an automatic equivalence relation whether it is of finite index. This can be achieved by constructing a synchronous transducer that accepts a set of representatives of the equivalence classes of $E_j$ (based on a length-lexicographic ordering). Then $E_j$ has finite index if and only if this set of representatives is finite, which can be decided in polynomial time.

It is unclear whether this approach can be used to obtain an exponential time upper bound for the recognizability test.\footnote{\cite{CCG06} mainly focused on decidability, and they agree that the proof as presented in that paper does not yield an exponential time upper bound.} One can construct a family $(R_n)_{n\in\mathbb{N}}$ of automatic binary relations $R_n$ defined by a deterministic synchronous transducer of size $\mathcal{O}(n^2)$ such that every synchronous transducer defining $E_1$ has size (at least) exponential in $n$ (\textit{cf.}\ Example \ref{example_rec_in_aut_lower_bound_construction_E}).
It is unclear whether it is possible to decide in polynomial time for such a transducer whether the equivalence relation it defines is of finite index. For this reason, we revisit the problem for finite words in Section~\ref{sec:rec-in-aut} and provide an exponential time upper bound for binary relations using a different approach.

\begin{example}\label{example_rec_in_aut_lower_bound_construction_E}
	We construct for any $n\in\mathbb{N}$ an automatic (binary) relation $R_n$ in terms of a deterministic synchronous transducer of size $\mathcal{O}(n^2)$ such that every synchronous transducer defining $E_1$ has size (at least) exponential in $n$.
	Let $n\in \mathbb{N}$ and $\Sigma := \{0,1\}$. Consider the relation
	\[R_n := \{(u\#v,t)\mid u,v,t\in\Sigma^n,~|t|_1 \leq 1,~\forall 0\leq i<n:t[i] = 1\rightarrow u[i] = v[i] \}\]
	That is, $R_n$ consists of tuples $(u\#v,t)$ where $u,v$ and $t$ are bit strings of length $n$ and $t$ contains at most one $1$.
	Moreover, the occurrence of this $1$ (if present) marks a position where $u$ and $v$ are equal.
	
	A deterministic synchronous transducer $\mathcal{A}_n$ can define $R_n$ with the help of two finite counters ranging over $\{0,\ldots,n\}$ as follows: the first counter measures the length of $u$ and $t$ and the second counter determines the position of the $1$-symbol in the second component $t$ (if present).
	Initially, both counters are increased by one on each transition of $\mathcal{A}_n$.
	Whether the second component is malformed (\textit{i.e.}\ contains two $1$'s), is verifiable with a single control bit in the state space.
	Also, $\mathcal{A}_n$ can remember the bit indicated by a $1$ in the second component with a single bit in the state space and stop the second counter (containing the correct position of the $1$).
	Up on reaching the separator $\#$ in the first component, $\mathcal{A}_n$ resets the first counter (assuming the input has been well-formed so far; otherwise, it rejects) and utilizes it to verify the length of $v$.
	Moreover, it decreases the second counter on each transition. If $0$ is reached it has found the position $i$ in $v$ marked by the second component and can compare $v[i]$ with $u[i]$ which has been saved in the state space.
	Finally, the case that the second counter does not stop (\textit{i.e.}\ there is no $1$ in the second component) can be handled with another control bit in the state space.
	$\mathcal{A}_n$ has to store both counters plus finitely many control/memory bits (whose number is independent of $n$) in the state space.
	Hence, $R_n$ is definable by a deterministic synchronous transducer of size $\mathcal{O}(n^2)$.
	
	It remains to show that every transducer defining $E_1$ has size exponential in $n$. For each input $u\#v\in \Sigma^n\{\#\}\Sigma^n$ it holds that\footnote{$u\oplus v$ shall denote the bitwise XOR operation on $u$ and $v$.}
	\[R|_{u\#v} = R(u\#v) = \{0^{i}10^{n-i-1}\mid 0\leq i< n,~ (u\oplus v)[i] = 0 \}\cup \{0^n\}.\]
	It suffices to show that every transducer recognizing
	\begin{multline*}
	E_1':=E_1\cap\big(\Sigma^n\{\#\}\Sigma^n\big)^2 =\\ \{(u\#v,u'\#v')\mid u,v,u',v'\in\Sigma^n\wedge\forall~ 0\leq i< n:  (u\oplus v)[i] = (u'\oplus v')[i] \}
	\end{multline*}
	has size exponential in $n$.
	This holds because $\big(\Sigma^n\{\#\}\Sigma^n\big)^2$ is definable by  a deterministic synchronous transducer of size $\mathcal{O}(n)$.
	Note that the existence of a synchronous transducer of sub-exponential size for $E_1$ would imply the existence of one for the intersection.
	Furthermore, observe that we cut out all pairs of malformed inputs $x\in (\Sigma\cup\{\#\})^*$ --- \textit{i.e.}\ with $R(x) = \emptyset$.
	For the sake of contradiction, suppose there is synchronous transducer $\mathcal{B}_n = (Q,\Sigma,\Sigma,q_0,\Delta,F)$ defining $E_1'$ with $|Q|<2^n$.
	By the pigeonhole principle there are $i\in\underline{n}$, $p\in Q$ and $u,u',x,x',v,v',y,y'\in\Sigma^n$ such that $(u\oplus u')[i] \neq (x\oplus x')[i]$ and
	\[\mathcal{B}_n: q_0\xrightarrow{u\#/u'\#}p\xrightarrow{v/v'}F\text{ as well as }\mathcal{B}_n: q_0\xrightarrow{x\#/x'\#}p\xrightarrow{y/y'}F. \]
	In particular, $u\#v$ and $u'\#v'$ as well as $x\#y$ and $x'\#y'$ are equivalent.
	Moreover, $x\#v$ and $x'\#v'$ are equivalent, too, since $\mathcal{B}$ permits an accepting run.
	W.l.o.g. we have that $(u\oplus u')[i] = 0$ and $(x\oplus x')[i] = 1$.
	That is, $u[i] = u'[i]$ and $x[i] \neq x'[i]$.
	Thus, $(u\oplus v)[i] = (u'\oplus v)[i] = (u'\oplus v')[i]$.
	We deduce $v[i] = v'[i]$.
	Then, similarly to the previous reasoning, we have that $(x\oplus v)[i] = (x\oplus v')[i] = (x'\oplus v')[i]$.
	But this yields $x[i] = x'[i]$ which is a contradiction to $(x\oplus x')[i] = 1$.
	Thus, every deterministic transducer defining $E_j'$ has size at least $2^n$ which proves our claim.
\end{example}

\subsection{From Indices of Equivalence Relations to Slenderness of Languages}\label{subsec:omega-rec-in-aut-infinite}
We now turn to the case of infinite words. The relation $E_j$ can be shown to be $\omega$-automatic, similarly to the case of finite words. However, it is not possible, in general, for a given $\omega$-automatic relation to define a set of representatives by means of a synchronous transducer, as shown by \cite{KL06}: There exists a binary $\omega$-automatic equivalence relation such that there is no $\omega$-regular set of representatives of the equivalence classes.

Here is how we proceed instead. The first step is similar to the approach of \cite{CCG06}:
We construct  synchronous transducers for the complements $\overline{E_j}$ of the equivalence relations $E_j$ in polynomial time (starting from a deterministic transducer for $R$).
We then provide a decision procedure to decide for a given transducer for $\overline{E_j}$ whether the index of $E_j$ is finite in doubly exponential time. This procedure is based on an encoding of ultimately periodic words by finite words.

First observe that a tuple in $\Sigma_1^\omega\times\ldots\times\Sigma_j^\omega$ can be seen as an infinite word over $\Sigma = \Sigma_1\times\ldots\times\Sigma_j$ (this is not the case for tuples over finite words, since the words may be of different length).
Hence, we can view each $E_j$ as a binary equivalence relation $E\subseteq \Sigma^\omega\times\Sigma^\omega$. For this reason, we only work with binary relations in the following.

We start by showing that for deciding whether $E$ has finite index it suffices to consider sets of ultimately periodic representatives $u_iv_i^\omega$ such that the periods $|v_i|$ and prefix lengths $|u_i|$ are the same, respectively, for all the representatives (Lemma~\ref{rec_in_aut_omega_lemma_2}).
In the second step $E$ is transformed into an automatic equivalence relation $E_\#$ over finite words using encodings of ultimately periodic words as finite words, where a word $uv^\omega$ is encoded by $u\#v$ as done by \cite{CNP93} (Definition~\ref{definition_E_finite_A} and Lemma~\ref{rec_in_aut_omega_lemma_3}).
Since $E_\#$ is an automatic relation over finite words, it is possible to obtain a finite automaton for a set of representatives of $E_\#$.
Finally, we reduce the decision problem whether $E$ has finite index to deciding slenderness (see Definition~\ref{def:slender} below) for polynomially many languages derived from the set of representatives of $E_\#$ (Lemmas~\ref{rec_in_aut_omega_lemma_4} \& \ref{rec_in_aut_omega_lemma_5}).
Therefore, by proving that deciding slenderness for (nondeterministic) finite automata is \textsc{NL}-complete (Lemma~\ref{lemma_slenderness_decidable}) we obtain our result.

\begin{definition}[\cite{PS95}] \label{def:slender}
	A language $L\subset\Sigma^*$ is slender if there exists a $k<\omega$ such that for all $\ell<\omega$ it holds that $|L\cap \Sigma^\ell|< k$.
\end{definition}

We now formalize the ideas sketched above.

\begin{lemma}\label{rec_in_aut_omega_lemma_2}
	Let $E\subseteq \Sigma^\omega\times\Sigma^\omega$ be an $\omega$-automatic equivalence relation.
	Then $E$ has not finite index if and only if for each $k> 0$ there are
	\[u_1,\ldots,u_k,~v_1,\ldots v_k\in \Sigma^* \text{ with } |u_i| = |u_j| \text{ and } |v_i|=|v_j|\text{ for all } 1\leq i\leq j\leq k  \]
	such that $(u_iv_i^\omega,u_jv_j^\omega)\notin E$ for all $1\leq i < j\leq k$.
\end{lemma}

To prove Lemma~\ref{rec_in_aut_omega_lemma_2} we first show the following, slightly weaker, version of it:
\begin{lemma}\label{rec_in_aut_omega_lemma_1}
	Let $E\subseteq \Sigma^\omega\times\Sigma^\omega$ be a $\omega$-automatic equivalence relation.
	Then $E$ has not finite index if and only if there are infinitely many (pairwise different) equivalence classes of $E$ containing an ultimately periodic representative.
\end{lemma}
\begin{proof}
	\underline{$\Leftarrow$}: Indeed, if $E$ has finite index then there are only finitely many equivalence classes of $E$ (containing an ultimately periodic word) by definition.
	
	\underline{$\Rightarrow$}:
	Suppose there are only finitely many equivalence classes of $E$ containing an ultimately periodic word.
	Let $C_1,\ldots,C_n$ be those equivalence classes and $u_1v_1^\omega,\ldots, u_nv_n^\omega$ be ultimately periodic words such that $u_iv_i^\omega$ is in $C_i$ for all $1\leq i\leq n$.
	Each $C_i$ is $\omega$-regular language, since it holds that
	\[ C_i = \dom(\{(w,u_iv_i^\omega)\in E \mid w\in\Sigma^\omega\}) = \dom(E\cap(\Sigma^\omega\times \{u_iv_i^\omega\})), \]
	$\omega$-automatic relations are closed under intersection and projection, and $\{u_i,v_i^\omega\}\times\Sigma^\omega$ is clearly $\omega$-automatic (in terms of automata, we can obtain an automaton for $C_i$ by fixing the second input tape of a synchronous transducer for $E$ to $u_iv_i^\omega$).
	
	Therefore, the (finite) union $C$ of the $C_i$ and its complement $\Sigma^\omega\setminus C$ are also $\omega$-regular.
	If $\Sigma^\omega\setminus C$ is non-empty, it contains an ultimately periodic word due to \cite{Buec62}.
	But this would contradict the assumption that the $C_i$ are all equivalence classes of $E$ containing an ultimately periodic word.
	Thus, the complement of $C$ is empty.
	It follows that $C_1,\ldots,C_n$ are all equivalence classes of $E$, and, hence, $E$ has finite index.
\end{proof}

\begin{proof}[of Lemma~\ref{rec_in_aut_omega_lemma_2}]
	\underline{$\Leftarrow$}:
	We prove the claim by contraposition.
	Suppose $E$ has finite index.
	Let $m_0:=\equiindex(E)<\omega$.
	Then for each collection of words $u_1,\ldots,u_m,~v_1,\ldots v_m\in \Sigma^*$ with $|u_i| = |u_j|=\ell$ and $|v_i|=|v_j|=p$ such that all $u_iv_i^\omega$ are pairwise non-equivalent, we have that $m\leq m_0$.
	Otherwise, there would be more than $m_0$ equivalence classes which is a contradiction.
	
	\underline{$\Rightarrow$}: Suppose $E$ has not finite index.
	Let $m> 0$. Due to Lemma \ref{rec_in_aut_omega_lemma_1} there are infinitely many pairwise non-equivalent (w.r.t.\ $E$) ultimately periodic words.
	Hence, we can pick ultimately periodic words $w_i = u_iv_i^\omega,~i\in \underline{m}$ which are pairwise non-equivalent. That is, $(u_iv_i^\omega,u_jv_j^\omega)\notin E$ for all indices $1\leq i<j\leq m$.
	We rewrite these ultimately periodic words such that they meet the conditions of the claim.
	W.l.o.g.\ we have that $|u_1|\geq |u_i|$ for all $i\in\underline{m}$.
	We define $u_1' := u_1$ as well as $v_1' := v_1$.
	Moreover, for each $i\in\underline{m}\setminus \{1\}$ let $p_i:=\lfloor\frac{|u_1|-|u_i|}{|v_i|}\rfloor$.
	Furthermore, consider the factorization $v_i = \hat{v}_i\tilde{v}_i$ where $|\hat{v}_i|=(|u_1|-|u_i|)\bmod |v_i|$. Note that $|u_1|-|u_i|\geq 0$, since $|u_1|\geq|u_i|$.
	Lastly, we define $u_i':= u_iv_i^{p_i}\hat{v}_i$.
	Then it holds that
	\begin{align*}
	|u_i'| &= |u_i|	&+&~p_i|v_i|										&+&~|\hat{v}_i|\\
	&= |u_i|		&+&~\lfloor\frac{|u_1|-|u_i|}{|v_i|}\rfloor|v_i|	&+&~((|u_1|-|u_i|)\bmod |v_i|)\\
	&= |u_i|		&+&~\big[|u_1|-|u_i| - ((|u_1|-|u_i|)\bmod |v_i|)\big]		&+&~((|u_1|-|u_i|)\bmod |v_i|)\\
	&= |u_i|		&+&~|u_1|-|u_i|&\\
	&= |u_1|		&&&
	\end{align*}
	Therefore, we have that $|u_i'|= | u_1 | = |u_j'|$ for all $1\leq i < j\leq m$.
	Moreover,
	\[ u_iv_i^\omega = u_iv_i^{p_i}\hat{v}_i\tilde{v}_i(\hat{v}_i\tilde{v}_i)^\omega = u_i'\tilde{v_i}(\hat{v}_i\tilde{v}_i)^\omega = u_i'(\tilde{v}_i\hat{v}_i)^\omega.\]
	Thus, by defining $v_i' := \tilde{v}_i\hat{v}_i$ we derive pairs $u_i',v_i'\in\Sigma^*$ such that $w_i = u_i'(v_i')^\omega$ and $|u_i'|=|u_j'|$ for all $1\leq i < j\leq m$.
	
	It remains to rewrite the $v_i'$ such that all periods have the same length.
	For that purpose, let $\ell:=\lcm(|v_1'|,\ldots,|v_k'|)$, $\ell_i = \frac{\ell}{|v_i'|}$, and define $v_i'':= (v_i')^{\ell_i}$ for all $i\in\underline{m}$.
	Then $u_i'(v_i')^\omega = u_i'(v_i'')^\omega = u_iv_i^\omega$ and $|v_i''| = |v_i'| \frac{\ell}{|v_i'|} = \ell$.
	Hence, the pairs $u_i',v_i''\in\Sigma^*$ satisfy the conditions of the claim.
\end{proof}

We proceed by transforming $E$ into an automatic equivalence relation $E_\#$ and showing that it is possible to compute in exponential time a synchronous transducer for it, given a synchronous B\"uchi transducer for $\overline{E}$.
\begin{definition}\label{definition_E_finite_A}
	Let $E\subseteq \Sigma^\omega\times\Sigma^\omega$ be an $\omega$-automatic equivalence relation. Furthermore, let $\Gamma := \Sigma\cup\{\#\}$ for a fresh symbol $\#\notin\Sigma$.
	Then the relation $E_\# \subseteq \Gamma^*\times\Gamma^*$ is defined by
	\[
	E_\# := \{(u\#v,x\#y)\mid u,v,x,y\in \Sigma^*, |u|=|x|, |v|=|y|, (uv^\omega,xy^\omega)\in E\}.
	\]
\end{definition}

\begin{lemma}
\label{rec_in_aut_omega_lemma_3}
	Let $E\subseteq \Sigma^\omega\times\Sigma^\omega$ be an $\omega$-automatic equivalence relation and $\mathcal{A}$ a synchronous B\"uchi transducer defining the complement $\overline{E}$ of $E$.
	Then, given $\mathcal{A}$, one can construct a synchronous transducer $\mathcal{A}_\#$ defining $E_{\#}$ in exponential time in the size of $\mathcal{A}$.
	In particular, $E_{\#}$ is an automatic relation and $\mathcal{A}_\#$ has size exponential in $\mathcal{A}$.
\end{lemma}
For the proof of Lemma \ref{rec_in_aut_omega_lemma_3}, we introduce the notion of transition profiles, which also play a central role in the original complementation proof for B\"uchi automata, as described, \textit{e.g.}, in \cite[Section~2]{Tho90}.
\begin{definition}[transition profile]\label{def:transition-profile}
	Let $\mathcal{A} = (Q,\Sigma,q_0,\Delta,F)$ be a (nondeterministic) B\"uchi automaton and $w\in\Sigma^*$.
	A \emph{transition profile} over $\mathcal{A}$ is a directed labeled graph $\tau = (Q,E)$
	where $E\subseteq Q\times\{1,F\}\times Q$.
	The transition profile $\tau(w) = (Q,E_w)$ induced by $w$ is the transition profile where $E_w$ contains an edge from $p$ to $q$ if and only if $p\xrightarrow{w}q$, and this edge is labeled with $F$ if and only if $p\smash{\xrightarrow[\smash{F}]{w}}q$.
	Finally, $TP(\mathcal{A}) := \{\tau(w)\mid w\in\Sigma^* \}$ denotes the set of all transition profiles over $\mathcal{A}$ induced by a word $w\in\Sigma^*$.
\end{definition}
It is well-known that for all words $v, w$ the transition profile $\tau(vw)$ is determined by the transition profiles $\tau(v)$ and $\tau(w)$.
In particular, $(TP(\mathcal{A}),\cdot)$ with $\tau(v)\cdot\tau(w) = \tau(vw)$ is a monoid with neutral element $\tau(\epsilon)$.
The following lemma is a simple observation that directly follows from the definition of transition profiles. 
\begin{lemma}[\cite{BLO12}]
	\label{lemma_transition_profile_acceptance}
	Let $\mathcal{A} = (Q,\Sigma,q_0,\Delta,F)$ be a B\"uchi automaton and $uv^\omega\in\Sigma^\omega$ be an ultimately periodic word.
	Then $uv^\omega\in L(\mathcal{A})$ if and only if there is a $p\in Q$ such that there is an edge from $q_0$ to $p$ in $\tau(u)$ and in $\tau(v)$ a cycle with an $F$ labeled edge is reachable from $p$.
\end{lemma}

\begin{proof}[of Lemma \ref{rec_in_aut_omega_lemma_3}]
	Let $\mathcal{A}$ be given by $\mathcal{A} = (Q,\Sigma,\Sigma,q_0,\Delta,F)$.
	We have to construct a synchronous transducer $\mathcal{A}_\#$ defining $E_{\#}$.
	Informally, on an input $(u,u')(\#,\#)(v,v')$ it works as follows.
	While reading $(u,u')$ the transducer $\mathcal{A}_\#$ computes the transition profile\footnote{Since $\mathcal{A}$ is a synchronous transducer we can view it as an Büchi automaton over the alphabet $\Sigma\times\Sigma$ which allows us to utilize transition profiles.} $\tau(u,u')$.
	After skipping $(\#,\#)$ it proceeds by computing the transition profile $\tau(v,v')$ while remembering $\tau(u,u')$.
	In the end, $\mathcal{A}_\#$ accepts if and only if for all states $p\in Q$ either in $\tau(v,v')$ no cycle with an edge labeled $F$\footnote{Recall that in transitions profiles (unlike transducers) edges may be labeled with $F$, \textit{cf.}\ Definition~\ref{def:transition-profile}.} is reachable from $p$ or there is no edge from $q_0$ to $p$ in $\tau(u,u')$.
	
	More formally, we define $\mathcal{A}_\# := (Q_\#,\Sigma,\Sigma,\tau(\varepsilon),\delta_\#,F_\#)$ where
	\begin{align*}
	Q_\# &=TP(\mathcal{A})\cup (TP(\mathcal{A})\times TP(\mathcal{A}))\text{ and,}\\
	F_\# &=\{(\tau,\tau')\in TP(\mathcal{A})^2\mid\forall p\in Q:
	\begin{cases}
	\text{in }\tau \text{ there is no edge from }q_0\text{ to  }p\text{, or}\\
	\text{in }\tau'\text{ no cycle with an edge labeled }F\\\hspace{7em}\text{is reachable from }p
	\end{cases}\}.
	\end{align*}
	The states $\tau\in TP(\mathcal{A})$ are used to read the $(u,u')$ prefix of the input while states $(\tau,\tau')$ are used to process the $(v,v')$ postfix.
	Thereby, $\tau$ is the current transition profile computed by $\mathcal{A}$ for $(u,u')$ and $\tau'$ is the current transition profile for $(v,v')$.
	Accordingly, the transition relation is defined as follows:
	\begin{align*}
	\Delta_\# := &\{(\tau,(a,b),\tau')\mid\tau,\tau'\in TP(\mathcal{A}), \tau' = \tau\cdot \tau(a,b),~ a,b\in\Sigma \}\\
	&\{(\tau,(\#,\#),(\tau,\tau(\varepsilon)))\mid \tau\in TP(\mathcal{A})\}\\
	&\{((\tau,\tau'),(a,b),(\tau,\tau''))\mid\tau,\tau'\in TP(\mathcal{A}), \tau'' = \tau'\cdot t(a,b),~ a,b\in\Sigma \}
	\end{align*}
	
	\underline{Complexity:} We have that $|Q_\#| = |TP(\mathcal{A})|+|TP(\mathcal{A})|^2\in \mathcal{O}(|TP(\mathcal{A})|^2)$.
	Furthermore, a transition profile can be described by a function $\tau: Q\times Q\rightarrow \{0,1,F\}$ --- \textit{i.e.}\ there is no edge, there is an edge  without label, or there is an edge labeled $F$ from $p$ to $q$ if $(p,q)$ is mapped to $0$, $1$, or $F$, respectively.
	Thus, $|TP(\mathcal{A})| = 3^{|Q|^2}$.
	In addition, given a transition profile $\tau$ the conditions in the definition of $F_\#$ and $\Delta_\#$ can be decided in polynomial time by a nested depth first search on $\tau$.
	Hence, $\mathcal{A}_\#$ can be computed in exponential time given $\mathcal{A}$.
	
	\underline{Correctness:} Obviously, $\mathcal{A}_\#$ rejects any malformed input pair (\textit{e.g.}\ if $u$ and $u'$ have different length) because no transitions are defined for the cases $(\#,a)$, $(a,\#)$, $(\varepsilon,a)$, $(a,\varepsilon)$, $a\in\Sigma$ (or $F_\#\cap TP(\mathcal{A}) = \emptyset$ in the case that no $\#$ occurs).
	On the other hand, consider a well-formed input pair $(u,u')(\#,\#)(v,v')$ with $|u|=|u'|$ and $|v| = |v'|$.
	Recall that $(TP(\mathcal{A}),\cdot)$ is a monoid.
	Hence, the run of $\mathcal{A}_\#$ on $(u,u')$ is unique and ends in $\tau(u,u')$ (the initial state is the neutral element $\tau(\varepsilon)$).
	Furthermore, like in the case of the prefix $(u,u')$ the run of $\mathcal{A}_\#$ on the suffix $(v,v')$ starting in $(\tau(u,u'),\tau(\varepsilon))$ is unique and ends in $(\tau(u,u'),\tau(v,v'))$.
	Thus, by the definition of $F_\#$, the transducer $\mathcal{A}_\#$ accepts $(u,u')(\#,\#)(v,v')$ if and only if $|u|=|u'|$, $|v| = |v'|$ and for all $p\in Q$ there is no edge from $q_0$ to $p$ in $\tau(u,u')$ or in $\tau(v,v')$ no cycle with an edge labeled $F$ is reachable from $p$.
	With Lemma \ref{lemma_transition_profile_acceptance}, it follows that $\mathcal{A}_\#$ accepts if and only if $|u|=|u'|$, $|v|=|v'|$ and $(uv^\omega,u'v'^\omega)\notin\overline{E}$.
	In conclusion, $R_*(\mathcal{A}_\#) = E_\#$.
\end{proof}

With a synchronous transducer for $E_\#$ at hand, we can compute a synchronous transducer defining a set of unique representatives of $E_\#$ similarly to the approach of \cite{CCG06} which we outlined in Section~\ref{subsec:omega-rec-in-aut-revision}, specifically to the step described in the paragraph following the proof of Lemma~\ref{lemma_restrictions_recognizable_finite_implies_rec_omega_automatic}.
For convenience, we will denote the set of representatives obtained by this construction by $L_\#(E)$ (although it is not unique in general).
We can now readjust Lemma \ref{rec_in_aut_omega_lemma_2} to $E_\#$ (or, more precisely, $L_\#(E)$).

\begin{lemma}\label{rec_in_aut_omega_lemma_4}
	Let $E\subseteq\Sigma^\omega\times\Sigma^\omega$ be an $\omega$-automatic equivalence relation.
	Then $E$ has finite index if and only if there is a $k<\omega$ such that for all $m,n>0: |L_{\#}(E)\cap\Sigma^n\{\#\}\Sigma^m|\leq k$.
\end{lemma}
\begin{proof}
	We prove both directions by contraposition.
	Suppose $E$ does not have finite index.
	We have to show that for all $k>0$ there are $m,n> 0$ such that $|L_{\#}(E)\cap\Sigma^n\{\#\}\Sigma^m|>k$ holds.
	Let $k>0$.
	Due to Lemma \ref{rec_in_aut_omega_lemma_2} there are $k+1$ many pairs $(u_i,v_i)\in \Sigma^*\times\Sigma^*$ with $|u_i| = |u_j| =: n$ and $|v_i| = |v_j| =: m$ for all $1\leq i < j\leq k+1$ such that $(u_i,u_j)(v_i,v_j)^\omega \notin E$.
	It follows that $(u_i\#v_i,u_j\#v_j)\notin E_\#$  for each $1\leq i < j\leq k+1$.
	W.l.o.g.\ we can choose the $(u_i,v_i)$ as the lexicographical smallest pairs with this property.
	We claim that $u_i\#v_i\in L_\#(E)$ for all $i\in\underline{k+1}$.
	Assume that there is a $i\in\underline{k+1}$ such that $u_i\#v_i\notin L_\#(E)$.
	Then there are words $x,y\in\Sigma^*$ such that $(x\#y,u_i\#v_i)\in E_\#$ and $x\#y<_\text{lex}u_i\#v_i$.
	In particular, $|x| = |u_i| = |u_j|$ and $|y| = |v_i| = |v_j|$ for all $j\in\underline{k+1}$.
	But then, $(x\#y,u_j\#v_j)\notin E_\#$ because $E_\#$ is an equivalence relation.
	This is a contradiction to the minimality (w.r.t. the lexicographical order) of $u_i\#v_i$.
	Hence,
	\begin{multline*}
	|L_\#(E)\cap \{u\#v\mid |u|=n,|v|=m \}| \geq |L_\#(E)\cap \{u_i\#v_i\mid 1\leq i\leq k+1 \}| = k+1>k.
	\end{multline*}
	
	On the contrary, assume that $\forall k>0~\exists m,n> 0: |L_{\#}(E)\cap\Sigma^n\{\#\}\Sigma^m|>k$ does hold.
	Again, let $k>0$.
	Then there are $m,n>0$ such that for each $L_{m,n} := L_\#(E)\cap \Sigma^n\{\#\}\Sigma^m$ it holds that $|L_{m,n}| > k$.
	Thus, there are pairwise different pairs $(u_i,v_i)$ such that $u_i\#v_i\in L_{m,n}$ for $1\leq i\leq k$.
	Moreover, by definition we have that $|u_i| = |u_j| = n$ and $|v_i|=|v_j| = m$ for all $1\leq i<j\leq k$.
	We claim that for each $i\neq j$ we have that $(u_iv_i^\omega,u_jv_j^\omega)\notin E$.
	Otherwise, there are $i,j$ such that $(u_i\#v_i,u_j\#v_j)\in E_\#$ and $(u_j\#v_j,u_i\#v_i)\in E_\#$ since $E_\#$ is symmetric.
	But then, because both $u_i\#v_i$ and $u_j\#v_j$ are in $L_\#(E)$, we have that $u_i\#v_i \not<_\text{lex}u_j\#v_j$ and $u_j\#v_j \not<_\text{lex}u_i\#v_i$. This is a contradiction.
	Therefore, we conclude that $E$ does not have finite index due to Lemma \ref{rec_in_aut_omega_lemma_2}.
\end{proof}

Note that the condition in Lemma \ref{rec_in_aut_omega_lemma_4} is similar to slenderness but not equivalent to the statement that $L_{\#}(E)$ is slender.
For instance, consider the language $L$ given by the regular expression $a^*\#b^*$. For any $m,n>0$ we have that $|L\cap \Sigma^n\{\#\}\Sigma^m| = |\{a^n\#b^m\}|\leq 1$. But $L$ is not slender: Let $\ell >0$.
Then $a^{\ell-1-i}\#b^{i}\in L\cap \Sigma^\ell$ for all $0\leq i < \ell$.
Hence, $|L\cap \Sigma^\ell|\geq \ell$ and, thus, $L$ cannot be slender.
However, the next result shows that there is a strong connection between the condition in Lemma \ref{rec_in_aut_omega_lemma_4} and slenderness.

\begin{lemma}\label{rec_in_aut_omega_lemma_5}
	Let $L$ be a language of the form $L = \bigcup_{(i,j)\in I} L_i\{\#\}L_j$ where $I\subset \mathbb{N}^2$ is a finite index set and $L_i,L_j\subseteq(\Sigma\setminus\{\#\})^*$ are non-empty regular languages for each pair $(i,j)\in I$.
	Then there is a $k<\omega$ such that for all $m,n\geq 0: |L\cap\Sigma^n\{\#\}\Sigma^m|\leq k$ if and only if for all $(i,j)\in I$ it holds that $L_i$ and $L_j$ are slender.
\end{lemma}

\begin{proof}
	\allowdisplaybreaks
	It holds that	\begin{align}
	&~\exists k									&&\forall m,n\geq 0
	&&
	&&|L\cap\Sigma^n\{\#\}\Sigma^m|\leq k
	\label{rec_in_aut_omega_lemma_5_eq_1}\\
	\Leftrightarrow
	&~\exists k									&&\forall m,n\geq 0
	&&
	&&|\bigcup_{(i,j)\in I}(L_i\{\#\}L_j\cap\Sigma^n\{\#\}\Sigma^m)|\leq k
	\label{rec_in_aut_omega_lemma_5_eq_2}\\
	\Leftrightarrow
	&~\exists k									&&\forall m,n\geq 0
	&&
	&&\sum_{(i,j)\in I}|L_i\{\#\}L_j\cap \Sigma^n\{\#\}\Sigma^m|\leq k
	\label{rec_in_aut_omega_lemma_5_eq_3}\\
	\Leftrightarrow
	&~\big(\exists k_{i,j}\big)_{(i,j)\in I}	&&\forall m,n\geq 0
	&&\bigwedge_{(i,j)\in I}
	&&|L_i\{\#\}L_j\cap \Sigma^n\{\#\}\Sigma^m|\leq k_{i,j}
	\label{rec_in_aut_omega_lemma_5_eq_4}\\
	\Leftrightarrow
	&~\big(\exists k_{i,j}\big)_{(i,j)\in I}	&&\bigwedge_{(i,j)\in I}
	&&\forall m,n\geq 0
	&&|L_i\{\#\}L_j\cap \Sigma^n\{\#\}\Sigma^m|\leq k_{i,j}
	\label{rec_in_aut_omega_lemma_5_eq_5}\\
	\Leftrightarrow
	&~\bigwedge_{(i,j)\in I}					&&\exists k_{i,j}
	&&\forall m,n\geq 0
	&&|L_i\{\#\}L_j\cap \Sigma^n\{\#\}\Sigma^m|\leq k_{i,j}
	\label{rec_in_aut_omega_lemma_5_eq_6}
	\end{align}
	Note that $\eqref{rec_in_aut_omega_lemma_5_eq_2}\Rightarrow\eqref{rec_in_aut_omega_lemma_5_eq_3}$ does hold since $I$ is finite.
	Furthermore, $\eqref{rec_in_aut_omega_lemma_5_eq_4}\Rightarrow\eqref{rec_in_aut_omega_lemma_5_eq_5}$ and $\eqref{rec_in_aut_omega_lemma_5_eq_5}\Rightarrow\eqref{rec_in_aut_omega_lemma_5_eq_6}$ do hold because $\forall$ distributes over $\wedge$ and the locality principle, respectively.
	
	Further on, we show that $\eqref{rec_in_aut_omega_lemma_5_eq_6}$ holds if and only if for all $(i,j)\in I$ it holds that $L_i$ and $L_j$ are slender.
	
	\underline{$\Leftarrow$}:	We prove the claim by contraposition.
	Suppose $L_i$ or $L_j$ for some $(i,j)\in I$ is not slender, say $L_i$ (for the case that $L_j$ is not slender the reasoning is analogous).
	Then $\forall k\exists m: |L_i\cap \Sigma^m| > k$.
	Let $k>0$ and $m$ such that $|L_i\cap \Sigma^m| > k$.
	Pick $v\in L_j\neq\emptyset$ and define $n:=|v|$.
	Note that by the choice of $n$ we have that $|L_j\cap\Sigma^n|\geq 1$.
	Clearly, it holds that $|L_i\{\#\}L_j| = |L_i||L_j|$.
	Moreover, since $L_i$ and $L_j$ do not contain any word with the letter $\#$ and both languages are non-empty by assumption, it follows that
	\[|L_i\{\#\}L_j\cap \Sigma^n\{\#\}\Sigma^m| = |L_i\cap \Sigma^n||L_j\cap \Sigma^m| > k\cdot 1\geq k.\]
	Thus, $\forall k\exists m,n\geq 0 |L_i\{\#\}L_j\cap \Sigma^n\{\#\}\Sigma^m| > k$.
	Hence, $\eqref{rec_in_aut_omega_lemma_5_eq_6}$ does not hold.
	
	\underline{$\Rightarrow$}: Let $(i,j)\in I$. By assumption $L_i$ and $L_j$ are slender.
	Thus, there are $k_i,k_j > 0$ such that $|L_i\cap \Sigma^n|\leq k_i$ and $|L_j\cap \Sigma^m|\leq k_j$ for all $m,n\geq 0$.
	It follows that for all $m,n\geq 0:$
	\[|L_i\{\#\}L_j\cap \Sigma^n\{\#\}\Sigma^m| = |L_i\cap \Sigma^n||L_j\cap \Sigma^m| \leq k_i k_j=:k.\]
	Hence, $\eqref{rec_in_aut_omega_lemma_5_eq_6}$ does hold, and thus, the lemma is proved.
\end{proof}

The last ingredient we need is the decidability of slenderness in polynomial time.
Lemma~\ref{lemma_slenderness_decidable} can be shown analogously to the proof given by \cite{Tao06} where it is shown that the finiteness problem for B\"uchi automata is \textsc{NL}-complete.
Indeed, there is a strong connection between these two problems which we shall briefly revisit in Subsection \ref{subsec:omega-rec-in-aut-slenderness-vs-finiteness}.
\begin{lemma}\label{lemma_slenderness_decidable}
	Deciding slenderness for (nondeterministic) finite automata is \textsc{NL}-complete.
\end{lemma}
\begin{proof}
	The proof of this lemma corresponds essentially to the proof given by \cite{Tao06} for \textsc{NL}-completeness of the finiteness problem for B\"uchi automata.
	However, there are some minor but critical technical differences.
	
	We prove that the non-slenderness problem, \textit{i.e.}\ whether for a given automaton $\mathcal{A}$ the language $L_*(\mathcal{A})$ is not slender, is \textsc{NL}-complete.
	Then the \textsc{NL}-complete\-ness of the slenderness problem follows immediately because \textsc{NL} = \textsc{coNL} (\textit{cf.}\ \cite{Sze88}).
	
	To show \textsc{NL}-hardness it suffices to provide a many-one reduction from the reachability problem for directed graphs which is complete for \textsc{NL}.
	Given a directed graph $\mathcal{G}$ and two nodes  $s,t$ of $\mathcal{G}$ we obtain an automaton $\mathcal{A}_\mathcal{G}$ over the alphabet $\{a,b\}$ by labeling each edge of $\mathcal{G}$ with $a$ and declaring $s$ and $t$ to be the initial state and the (sole) accepting state, respectively.
	Furthermore, we add two transitions $(s,a,s)$ and $(s,b,s)$. Then $\mathcal{A}_\mathcal{G}$ recognizes the non-slender language $\{a,b\}^*L$ for some $L\subseteq \{a\}^*$ if and only if $t$ is reachable from $s$ in $\mathcal{G}$, and, otherwise, $\emptyset$.
	
	Let $\mathcal{A} = (Q,\Sigma,q_0,\Delta,F)$ be the given automaton.
	We claim that $L_*(\mathcal{A})$ is not slender if and only if there are $q,p_1,p_2\in Q$ and $f_1,f_2\in F$ such that
	\begin{enumerate}
		\item $q_0\xrightarrow{w_0} q$ and $q\xrightarrow{w} q$ for some $w_0\in\Sigma^*$ and $w\in\Sigma^+$,\label{lemma_slenderness_decidable_item1}
		\item there are $u_1,u_2\in\Sigma^+$ with $u_1[i]\neq u_2[i]$ for an index $i\leq \min(|u_1|,|u_2|)$, $q\xrightarrow{u_1}p_1$, and $q\xrightarrow{u_2}p_2$, and\label{lemma_slenderness_decidable_item2}
		\item there are $w_1,w_2\in\Sigma^+, v_1,v_2\in\Sigma^*: p_1\xrightarrow{w_1}p_1\xrightarrow{v_1}f_1$ and  $p_2\xrightarrow{w_2}p_2\xrightarrow{v_2}f_2$.\label{lemma_slenderness_decidable_item3}
	\end{enumerate}
	Suppose our claim holds.
	Then membership in \textsc{NL} because the conditions can easily verified by a nondeterministic logspace Turing machine (all conditions boil down to reachability, $u_1[i]\neq u_2[i]$ can be asserted on the fly in a parallel search).
	
	It remains to prove the claim.
	Suppose conditions \ref{lemma_slenderness_decidable_item1},\ref{lemma_slenderness_decidable_item2}, \ref{lemma_slenderness_decidable_item3} hold.
	Then either $u_1$ or $u_2$ is not a prefix of $w^\omega$, say w.l.o.g.\ $u_1$.
	Furthermore, we can assume that $|w_1| = |w|$.
	Otherwise, by repeating each word until the least common multiple of their lengths is reached we get words satisfying this property.
	Hence, for all $i,j$ the labelings of the accepting runs
	\[ q_0\xrightarrow{w_0}q\xrightarrow{w^i}q\xrightarrow{u_1}p_1\xrightarrow{w_1^j}p_1\xrightarrow{v_1}f_1 \]
	are pairwise different.
	Thus, $L_*(\mathcal{A})$ is not slender (for all solutions of $i+j = n$ for a fixed $n$ a unique word in $L_*(\mathcal{A})$ is obtained and all these words have the same length).
	
	On the contrary, suppose $L_*(\mathcal{A})$ is not slender.
	Consider the set of states
	\[P := \{q\in Q \mid \exists f\in F~\exists w\in\Sigma^+: q_0\rightarrow q\xrightarrow{w}q\rightarrow f  \}.\]
	If $P$ is empty then $L_*(\mathcal{A})$ is finite, and, thus, slender which is a contradiction.
	Assume for the sake of contradiction that for no $q\in P$ there are $p_1,p_2,f_1,f_2$ as above satisfying, together with $q$, the conditions \ref{lemma_slenderness_decidable_item1},\ref{lemma_slenderness_decidable_item2}, \ref{lemma_slenderness_decidable_item3}.
	Let $q\in P$ and $f\in F, w_q\in\Sigma^+$ be witnessing the membership of $q\in P$.
	By choosing $p_2 := q, f_2:=f$ and $u_2:= w_q$ we have that there is no $u_1$ which is not a prefix of $w^\omega$ and leads from $q$ to a productive state $p_1$  that is reachable from itself (via a non-empty word $w_2$).
	Let $\mathcal{A}_q$ be the automaton $\mathcal{A}$ with initial state $q$.
	We conclude that $L_*(\mathcal{A}_q) \subseteq \{w_q \}^*Z_q$ where $Z_q$ is a finite language.
	Moreover, since $P$ contains all productive states with a self-loop it follows that, up to finitely many words, $L_*(\mathcal{A})\subseteq \bigcup_{q\in P}X_q\{w_q\}^*Z_q$.
	Finally, observe that $X_q$ can be assumed to be finite.
	Otherwise, $q$ is reachable from another state in $q'\in P$ and, thus, $L_*(\mathcal{A}_q)Z_q$ is subsumed by $L_*(\mathcal{A}_{q'})Z_{q'}$.
	Then it is immediate that $\bigcup_{q\in P}X_q\{w_q\}^*Z_q$ is a slender language.
	It follows that $L_*(\mathcal{A})$ is slender which is a contradiction.
\end{proof}

Finally, we can combine our results to obtain the main result of this section.
Firstly, we state our approach to check whether an automatic equivalence has finite index and, afterwards, join it with the approach of \cite{CCG06}.
\begin{theorem}\label{lemma_rec_in_aut_omega_finite_index_decidable}
	Let $E\subseteq\Sigma^\omega\times\Sigma^\omega$ be an $\omega$-automatic equivalence relation and $\mathcal{A}_\#$ be a (nondeterministic) synchronous transducer defining $E_\#$.
	Then it is decidable in single exponential time whether $E$ has finite index.
\end{theorem}
\begin{proof}
		Let $<_\text{lex}$ denote some (fixed) lexicographical ordering on $(\Sigma\cup\{\#\})^*$.
		We claim that the following algorithm decides, given a synchronous transducer $\mathcal{A}_\#$ defining $E_\#$, whether $E$ has finite index in exponential time.
	\begin{enumerate}
		\item Construct a synchronous transducer defining $E_\#^< := \{(u\#v,u'\#v')\in E_\# \mid u\#v <_\text{lex} u'\#v'\}$.
		\item Project $E_\#^<$ to the second component and obtain a transducer defining
		\begin{multline*}
		P_\# := \{u\#v\in \Sigma^*\{\#\}\Sigma^*\mid \exists u',v'\in \Sigma^*: u'\#v' <_\text{lex} u\#v \wedge (u'\#v',u\#v)\in E_\# \}.
		\end{multline*}
		\item Construct an automaton $\mathcal{B}_\# = (Q,\Sigma\cup\{\#\},\Delta,q_0,F)$ for $L_\#(E) = \overline{P_\#}\cap \Sigma^*\{\#\}\Sigma^*$.
		\item Construct automata for the factors $L_{q_0p}$ and $L_{qF}$ of the decomposition \[L_\#(E) = \bigcup_{(p,\#,q)\in\Delta} L_{q_0p}\{\#\}L_{qF}, \] of $L_\#(E)$ where \[L_{q_0p} := \{u\in\Sigma^*\mid \mathcal{B}_\#:q_0\xrightarrow{u}p \}\text{ and }L_{qF} := \{v\in\Sigma^*\mid \mathcal{B}_\#:q\xrightarrow{v}F \}.\]
		\item For each $(p,\#,q)$ such that $L_{q_0p}\neq \emptyset$ and $L_{qF}\neq \emptyset$,
		check if $L_{q_0p}$ and $L_{qF}$ are slender.
		If all checked languages are slender return yes ($E$ has finite index), otherwise no.
	\end{enumerate}

	\underline{Complexity:} Obtaining a transducer for $E_\#^<$ given $\mathcal{A}_\#$ is immediate.
	It can be checked on the fly by a modification of $\mathcal{A}_\#$ that rejects once it encounters a tuple $(x,y)$ of letters witnessing $u\#v\not <_\text{lex}u'\#v'$.
	Note that the condition $|u\#v| = |u'\#v'|$ is already verified by $\mathcal{A}_\#$.
	Projection and intersection of synchronous transducers can be achieved in polynomial time while the complementation of a synchronous transducer is achievable in exponential time.
	Hence, $\mathcal{B}_\#$ is exponential in the size of $\mathcal{A}_\#$.
	Further on, automata for $L_{q_0p}$ and $L_{qF}$ can easily be obtained from $\mathcal{B}$ in polynomial time.
	Furthermore, there are only polynomial many --- to be more precise, at most $|Q|^2$ many --- such languages $L_{q_0p}$ and $L_{qF}$ and emptiness as well as slenderness in the last step can be checked in polynomial time due to Lemma \ref{lemma_slenderness_decidable}.
	All in all, the given decision procedure runs in single exponential time.
	
	\underline{Correctness:}
	Indeed, $\mathcal{B}_\#$ defines a set of representatives of $E_\#$: it accepts precisely the words $u\#v$ such that there is no lexicographically smaller word $u'
	\#v'$ which is equivalent to $u\#v$ (\textit{cf.}\ \cite{CCG06}).
	By Lemma \ref{rec_in_aut_omega_lemma_4} $E$ has finite index if and only if there is a $k<\omega$ such that for all $m,n>0: |L_{\#}(E)\cap\Sigma^n\{\#\}\Sigma^m|\leq k$.
	If $L_{q_0p} = \emptyset$ or $L_{qF} = \emptyset$ for some $(p,\#,q)\in\Delta$ the segment $L_{q_0p}\{\#\}L_{qF}$ can be removed from the union $L_\#(E) = \bigcup_{(p,\#,q)\in\Delta} L_{q_0p}\{\#\}L_{qF}$ without altering the language.
	The remaining union (\textit{i.e.}\ with all those segments removed) satisfies the condition of Lemma \ref{rec_in_aut_omega_lemma_5}.
	It follows that $E$ has finite index if and only if there is a $k<\omega$ such that for all $m,n>0: |L_{\#}(E)\cap\Sigma^n\{\#\}\Sigma^m|\leq k$ if and only if for each $(p,\#,q)$ such that $L_{q_0p}\neq \emptyset$ and $L_{qF}\neq \emptyset$ it holds that $L_{q_0p}$ and $L_{qF}$ are slender.
	Thus, the decision procedure is correct.
\end{proof}

\begin{theorem}\label{theorem_rec_in_aut_omega_decidable}
	Given a complete deterministic synchronous parity transducer $\mathcal{A}$ it is decidable in double exponential time whether $R_\omega(\mathcal{A})$ is $\omega$-recognizable.
\end{theorem}
\begin{proof}
	Due to \cite{CCG06} we can obtain synchronous B\"uchi transducers for the $\omega$-automatic equivalence relations $\overline{E_j}$ w.r.t.\ $R:=R_\omega(\mathcal{A})$ for all $1\leq j\leq k$ in polynomial time.
	For the sake of completeness we will briefly sketch the construction.
	Let $R' := \{ (u_{j+1},\ldots,u_k,u_1,\ldots,u_j)\mid (u_1,\ldots,u_k)\in R\}$.
	That is, $R'$ is obtained from $R$ by swapping the first $j$ entries and the $k-j$ last entries of each tuple.
	Clearly, $R'$ is $\omega$-automatic.
	Hence, we can obtain transducers for $\overline{R}$ and $\overline{R'}$ in polynomial time, since $R$ is given by a \emph{complete} synchronous deterministic parity transducer.
	Furthermore, let $\circ_j$ be the composition operation on $k$-ary relations linking the last $k-j$ entries of a tuple  with the first $k-j$ entries.
	That is, if $(u_1,\ldots,u_j,w_{j+1},w_k)\in S$ and $(w_{j+1},w_k,v_1,\ldots,v_j)\in T$  then $(u_1,\ldots,u_j,v_1,\ldots,v_j)\in S\circ_j T$.
	Finally, observe that $\overline{E_j} = R\circ_j \overline{R'}\cup \overline{R}\circ_j R'$.
	Note that this step is achievable in polynomial time, since we are constructing a transducer for $\overline{E_j}$ and not $E_j$ (the transducers defining the compositions $R\circ_j \overline{R'}$ and $\overline{R}\circ_j R'$ are nondeterministic).
	
	Further on, observe that we can understand an equivalence relation $E_j$ over $\aryomegarelation{j}$ as an equivalence relation over $\Sigma^\omega$ with $\Sigma^\omega =\big(\Sigma_1\times\ldots\times\Sigma_j\big)^\omega \approx \aryomegarelation{j}$.
	Fix $j\in\underline{k}$.
	Then a transducer that defines $E_\#$ w.r.t.\ $E:=E_j$ is constructible in single exponential time due to Lemma \ref{rec_in_aut_omega_lemma_3} given a transducer for $\overline{E_j}$.
	By Lemma~\ref{lemma_restrictions_recognizable_finite_implies_rec_omega_automatic}, we have that $R$ is recognizable if and only if all $E_j$ have finite index.
	The latter is decidable for $E:=E_j$ in single exponential time given a transducer for $E_\#$  due to Theorem \ref{lemma_rec_in_aut_omega_finite_index_decidable}.
	
	Thus, it is decidable in double exponential time whether $R_\omega(\mathcal{A})$ is $\omega$-re\-cog\-niz\-able.
\end{proof}

\subsection{Slenderness vs.\ Finiteness}\label{subsec:omega-rec-in-aut-slenderness-vs-finiteness}
As mentioned above, we state the connection between the slenderness problem for regular languages and the finiteness problem for Büchi automata.
Recall that the algorithm given by \cite{CCG06} for deciding recognizability of automatic relations checks in the end whether the (regular) set of representatives of an equivalence relation is finite.
The decision procedure given in Theorem \ref{theorem_rec_in_aut_omega_decidable} for recognizability of $\omega$-automatic relations instead checks for slenderness.
We say that an automaton is trimmed w.r.t.\ the Büchi condition if for each state of the automaton there is a non-empty word leading to an accepting state.
\begin{lemma}\label{lemma_slenderness_vs_finiteness}
               Let $\mathcal{A}$ be an automaton trimmed w.r.t.\ the Büchi condition. Then $L_\omega(\mathcal{A})$ is finite if and only if $L_*(\mathcal{A})$ is slender.
\end{lemma}

\begin{proof}
	Suppose $L_\omega(\mathcal{A})$ is finite.
	Since $\mathcal{A}$ is trimmed w.r.t.\ the Büchi condition, each word $w\in L_*(\mathcal{A})$ is a prefix of some word in $L_\omega(\mathcal{A})$ because each (finite) run of $\mathcal{A}$ can be extended to an infinite run visiting infinitely many accepting states.
	Thus, for each $\ell<\omega$ there are at most  $k:= |L_\omega(\mathcal{A})| < \omega$ many words of length $\ell$.
	Hence, $L_*(\mathcal{A})$ is slender.

    Assume $L_\omega(\mathcal{A})$ is infinite and let $k<\omega$.
    It suffices to show that there are more than $k$ pairwise different words of the same length $\ell$ in $L_*(\mathcal{A})$.
    Then $L_*(\mathcal{A})$ cannot be slender.
    Since $\mathcal{A}$ has only finitely many states and $L_\omega(\mathcal{A})$ is infinite we can find an $\ell<\omega$, a state $q$ of $\mathcal{A}$, and $k+1$ infinite words $\alpha_0,\ldots,\alpha_k$ in $L_\omega(\mathcal{A})$ such that for $0\leq i\leq k$ the prefixes $\alpha_i[0,\ell]$ are pairwise different and an accepting run of $\mathcal{A}$ on $\alpha_i$ is in state $q$ after processing the prefix $\alpha_i[0,\ell]$.
    Let $w$ be a finite word leading from $q$ into some accepting state of $\mathcal{A}$.
    Note that $w$ exists because $\mathcal{A}$ is trimmed w.r.t.\ the Büchi condition.
    Then the words $\alpha_i[0,\ell]w$ are all pairwise different but of the same length and in $L_*(\mathcal{A})$.
\end{proof}
 \section{Deciding Recognizability of Automatic Relations}\label{section_visibly_binary_automatic} \label{sec:rec-in-aut}
In Section~\ref{subsec:omega-rec-in-aut-revision} we have sketched the approach presented by \cite{CCG06} for deciding recognizability of an automatic relation. In this section we revisit the problem to obtain an exponential time upper  bound for the case of binary relations. The procedure is based on a reduction to the regularity problem for  \textsc{VPA}s (Lemma~\ref{lemma_rec_in_aut_visibly_construction}). The other main contribution in this section, which is interesting on its own, is a polynomial time algorithm to solve the regularity problem for \textsc{DVPA}s. We start by describing the regularity test.

\subsection{Deciding Regularity for Deterministic Visibly Pushdown Automata}
We start by briefly discussing why the polynomial time regularity test
for visibly pushdown processes as presented by \cite{Srb06} does not
imply our result. The model used by \cite{Srb06} cannot use
transitions that cause a pop operation when the stack is empty. One
can try to circumvent this problem by introducing new internal symbols
that simulate pop-operations on the empty stack: For each $r \in
\Sigma$ introduce a new internal symbol $a_r$, and modify the DVPA
such that it can read $a_r$ instead of $r$ when the stack is empty (we
do not detail such a construction because it is straight
forward). This yields a DVPA without pop-operations on the empty
stack. However, this operation changes the accepted language, and this
change does not preserve regularity, in general. To see this, consider
the following example with one call symbol $c$ and one return symbol
$r$. The VPA has two states $q_c$ and $q_r$, where $q_c$ is initial,
and both states are final. The stack alphabet is $\{\gamma,\bot\}$,
and the transitions are $(q_c,c,q_c,\gamma)$, $(q_c,r,\gamma,q_r)$,
$(q_r,r,\gamma,q_r)$, $(q_r,r,\bot,q_r)$. This DVPA accepts the
regular language $c^*r^*$. Obviously, there is no DVPA that accepts
the same language without pop-operations on the empty stack. The
transformation into a DVPA without pop-operations on the empty stack,
as described above, introduces a new internal symbol $a_r$ and results
in a DVPA accepting all words of the form $c^nr^m$ with $m \le n$, or
of the form $c^nr^na_r^*$. This language is not regular anymore,
showing that such a transformation cannot be used in the context of a
regularity test.

We now proceed with the description of our polynomial time regularity
test for DVPAs. It is based on the following result, which states that
in a DVPA accepting a regular language, all configurations, that only
differ ``deep inside'' the stack, are equivalent.

\begin{lemma}\label{lemma_rec_in_aut_visibly_core_lemma}
	Let $\mathcal{P}$ be a \textsc{DVPA} with $n$ states. Then $L(\mathcal{P})$ is regular if and only if all pairs $(p,\alpha\beta), (p,\alpha\beta')$ of reachable configurations of $\mathcal{P}$ with $|\alpha| \geq n^3+1$ are $\mathcal{P}$-equivalent.
\end{lemma}
\begin{proof}
	Let $\mathcal{P} = (P,\Sigma,\Gamma,p_0,\bot,\Delta,F)$ be a deterministic visibly pushdown automaton. Furthermore, we set $n:=|P|$ and $m:=n^3+1$.
	
	\underline{$\Leftarrow$}:
	Suppose all pairs $(p,\alpha\beta), (p,\alpha\beta')$ of reachable configurations of $\mathcal{P}$ with $|\alpha| \geq m$ are $\mathcal{P}$-equivalent.
	Let $\mathfrak{C}\subseteq (P\times(\Gamma\setminus\{\bot\})^*\{\bot\})$ be the set of reachable configurations of $\mathcal{P}$.
	Then $\mathfrak{C}_{/\approx_\mathcal{P}}$ is finite, since each reachable configuration of the form $(p,\alpha\beta)$ is $\mathcal{P}$-equivalent to a reachable configuration $(p,\alpha\beta')$ with $|\beta'|\leq|\beta|$ minimal.
	Therefore, $L(\mathcal{P})$ is regular.
	A witnessing finite automaton is the canonical quotient automaton over the set of all reachable configurations given by 
	\[\mathcal{A}_{\approx_\mathcal{P}} = (\mathfrak{C}_{/\approx_\mathcal{P}},\Sigma,[q_0,\bot]_{\approx_\mathcal{P}}, \Delta_{/\approx_\mathcal{P}},F_{/\approx_\mathcal{P}})\] where $F_{/\approx_\mathcal{P}} = \{[p,\alpha]_{\approx_\mathcal{P}}\in \mathfrak{C}_{/\approx_\mathcal{P}}\mid p\in F \}$ and $\Delta_{\approx_\mathcal{P}}$ contains a transition $([p,\alpha]_{\approx_\mathcal{P}},a,[q,\beta]_{\approx_\mathcal{P}})$ if and only if $\mathcal{P}$ can proceed from $(p,\alpha)$ to some $(q',\beta')\in [q,\beta]_{\approx_\mathcal{P}}$ via $a$.
	
	\underline{$\Rightarrow$}:
	Suppose $L(\mathcal{P})$ is regular.
	Then there is a complete deterministic automaton $\mathcal{A}$ with state set $S$ defining $L(\mathcal{P})$.
	For the sake of contradiction, assume there are two reachable configurations $(p,\alpha\beta),(p,\alpha\beta')$ with $|\alpha|\geq m$ that are not $\mathcal{P}$-equivalent (obviously, this implies that $\beta\neq\beta'$).
	We claim that for each $\ell\in\mathbb{N}$ there is a pair $(p_\ell,\alpha_\ell),(p_\ell,\beta_\ell)$ of reachable, non-equivalent configurations such that for each $x\in L(p_\ell,\alpha_\ell)\triangle L(p_\ell,\beta_\ell)$ we have that $|x|>\ell$. In other words, there are configurations that can only be separated by words of length at least $\ell$ for each $\ell\in\mathbb{N}$.
	We postpone the proof of this claim and show that it is a contradiction to $L(\mathcal{P})$ being regular first.
	Let $u,v$ be words witnessing the reachability of $(p_\ell,\alpha_\ell)$ and $(p_\ell,\beta_\ell)$ for some $\ell\in\mathbb{N}$.
	Furthermore, let $s_u,s_v\in S$ be the unique states that are reached by $\mathcal{A}$ reading $u$ and $v$, respectively (starting in the initial state).
	Finally, w.l.o.g.\ pick a word $x_\ell \in L(p_\ell,\alpha_\ell)\setminus L(p_\ell,\beta_\ell)$ of minimal length.
	Note that $|x_\ell|\geq \ell$.
	Then $\mathcal{A}$ ends up in an accepting state reading $x_\ell$ starting from $s_u$ but in a non-accepting state reading $x_\ell$ starting from $s_v$ because $L(\mathcal{A}) = L(\mathcal{P})$.
	But then there is a word $y_\ell$ of length at most $|S|^2$ such that $\mathcal{A}$ ends up in an accepting state reading $y_\ell$ starting from $s_u$ but in a non-accepting state reading $y_\ell$ starting from $s_v$.
	Moreover, $y_\ell \in L(p_\ell,\alpha_\ell)\setminus L(p_\ell,\beta_\ell)$ because $L(\mathcal{A}) = L(\mathcal{P})$ and both automata are deterministic.
	This is a contradiction to the choice of $x_\ell$ for a sufficient large $\ell$.
	
	It remains to construct the configurations $(p_\ell,\alpha_\ell)$ and $(p_\ell,\beta_\ell)$.
	Let $u,v\in\Sigma^*$ be words witnessing the reachability of $(p,\alpha\beta)$ and $(p,\alpha\beta')$, respectively.
	Furthermore, let $w\in\Sigma^*$ be a witness for the non-equivalence of these configurations.
	W.l.o.g.\ $w\in L(p,\alpha\beta)\setminus L(p,\alpha\beta')$.
	That is, $uw\in L(\mathcal{P})$ but $vw\notin L(\mathcal{P})$ because $\mathcal{P}$ is deterministic.
	Since $\alpha$ is on the top of the stack of both configurations and the state component is the same, we have that $|w|\geq |\alpha|\geq m$ ($\mathcal{P}$ has to pop $\alpha$ from the stack while reading $w$ which requires $m$ return symbols in $w$; otherwise, the runs cannot differ).
	More precisely, there are well-matched words $w_1,\ldots, w_{m}$ and $r_1,\ldots,r_{m}$ such that $w=w_mr_m\ldots w_1r_1w'$ for some $w'\in\Sigma^*$.
	Similarly, $u=u'c_1^uu_2\ldots u_mc_m^uu_{m+1}$ and $v=v'c_1^vv_2\ldots v_mc_m^vv_{m+1}$ where the $u_i,v_i$ are well-matched words, $c_i^u,c_i^v$ are call symbols and $u',v'$ are words responsible for the lower stack contents $\beta$ and $\beta'$, respectively.
	All in all, the runs of $\mathcal{P}$ on $uw$ and $vw$ have the following shape ($\alpha = \gamma_n\ldots \gamma_1$, $f\in F$, $e\notin F$):\\
	\resizebox{\textwidth}{!}{		\begin{tikzpicture}
		\matrix(m)[matrix of math nodes, row sep=2em, column sep=2em, text height=1.5ex, text depth=0.25ex,ampersand replacement=\&]
		{(p_0,\bot) \& (p_1,\beta) \& (p_1',\gamma_1\beta) \& (p_2,\gamma_1\beta) \& \ldots \& (p'_m,\gamma_n\ldots\gamma_1\beta)  \&  (p,\alpha\beta)  \\
			\phantom{(p_0,\bot)}			\& (q_{m},\alpha\beta) \& (q_{m}',\gamma_{m-1}\ldots\gamma_1\beta) \& \ldots \& (q_1,\gamma_1\beta) \& (q_1',\beta) \& (f,\lambda)\\
			(p_0,\bot) \& (s_1,\beta') \& (s_1',\gamma_1\beta') \& (s_2,\gamma_1\beta') \& \ldots \& (s'_m,\gamma_n\ldots\gamma_1\beta')  \&  (p,\alpha\beta')  \\
			\phantom{(p_0,\bot)}			\& (q_{m},\alpha\beta') \& (q_{m}',\gamma_{m-1}\ldots\gamma_1\beta') \& \ldots \& (q_1,\gamma_1\beta') \& (q_1',\beta') \& (e,\rho)\\
		};
		\path[dashed]
		(m-1-2) edge node {} (m-2-5.north)
		(m-1-3) edge node {} (m-2-6.north)
		(m-3-2) edge node {} (m-4-5.north)
		(m-3-3) edge node {} (m-4-6.north)
		;
		\path[->,thick]
		(m-1-1) edge node[above] {$u'$} (m-1-2)
		(m-1-2) edge node[above] {$c_1^u$} (m-1-3)
		(m-1-3) edge node[above] {$u_2$} (m-1-4)
		(m-1-4) edge node[above] {$c_2^u$} (m-1-5)
		(m-1-5) edge node[above] {$c_m^u$} (m-1-6)
		(m-1-6) edge node[above] {$u_{m+1}$} (m-1-7)
		
		(m-2-1) edge node[above] {$w_m$} (m-2-2)
		(m-2-2) edge node[above] {$r_m$} (m-2-3)
		(m-2-3) edge node[above] {$w_{m-1}$} (m-2-4)
		(m-2-4) edge node[above] {$w_1$} (m-2-5)
		(m-2-5) edge node[above] {$r_1$} (m-2-6)
		(m-2-6) edge node[above] {$z'$} (m-2-7)
		
		(m-3-1) edge node[above] {$v'$} (m-3-2)
		(m-3-2) edge node[above] {$c_1^v$} (m-3-3)
		(m-3-3) edge node[above] {$v_2$} (m-3-4)
		(m-3-4) edge node[above] {$c_2^v$} (m-3-5)
		(m-3-5) edge node[above] {$c_m^v$} (m-3-6)
		(m-3-6) edge node[above] {$v_{m+1}$} (m-3-7)
		
		(m-4-1) edge node[above] {$w_m$} (m-4-2)
		(m-4-2) edge node[above] {$r_m$} (m-4-3)
		(m-4-3) edge node[above] {$w_{m-1}$} (m-4-4)
		(m-4-4) edge node[above] {$w_1$} (m-4-5)
		(m-4-5) edge node[above] {$r_1$} (m-4-6)
		(m-4-6) edge node[above] {$z'$} (m-4-7)
		;
		\end{tikzpicture}}\\
	Observe that the symbol pushed by the transitions originating in $p_i$ or $s_i$ is popped from the transition originating in $q_i$.
	Since $\mathcal{P}$ has $n$ states, there are $n^3$ possible valuations for a triple $(p_i,s_i,q_i)$.
	On the other hand, there are $m=n^3+1$ many triples $(p_i,s_i,q_i)$ in the outlined run.
	It follows that there are indices $1\leq i< j\leq m$ such that $(p_i,s_i,q_i) = (p_j,s_j,q_j)$.
	Let $\ell > 0$. By repeating the path fragments identified by $i$ and $j$ we obtain that the configurations $(p,\alpha_\ell\beta)$ and $(p,\alpha_\ell\beta')$ with $\alpha_\ell:= \gamma_m\ldots \gamma_{j}(\gamma_{j-1}\ldots\gamma_{i})^\ell \gamma_{i-1}\ldots\gamma_1$ are reachable.
	The reachability is witnessed by the words
	\begin{multline*}
	u_\ell = u'c_1^uu_2\ldots u_{i}(c_{i}^uu_{i+1}\ldots u_{j})^\ell c_j^uu_{j+1}\ldots u_{m+1}\\\text{and }v_\ell = v'c_1^vv_2\ldots v_{i}(c_{i}^vv_{i+1}\ldots v_{j})^\ell c_j^vv_{j+1}\ldots v_{m+1}.
	\end{multline*}
	Moreover, $(p,\alpha_\ell\beta)$ and $(p,\alpha_\ell\beta')$ are not $\mathcal{P}$-equivalent because the word
	\[w_mr_m\ldots w_{j}r_{j}w_{j-1}(r_{j-1}\ldots r_{i}w_{i-1})^\ell r_{i-1}\ldots r_1z'\] separates them.
	We conclude the proof by the observation that $(p,\alpha_\ell\beta)$ and $(p,\alpha_\ell\beta')$ cannot be separated by any word of length less than $\ell$, since $|\alpha_\ell|\geq\ell$.
\end{proof}

\begin{theorem}\label{theorem_rec_in_aut_visibly_decidable}
	It is decidable in polynomial time whether a given \textsc{DVPA} defines a regular language.
\end{theorem}

In the proof of Theorem \ref{theorem_rec_in_aut_visibly_decidable} we will make extensive use of the following well-known result for pushdown systems:
\begin{proposition}[\cite{BEM97}]\label{proposition_rec_in_aut_visibly}
	Let $\mathcal{P} = (P,\Sigma,\Gamma,p_0,\bot,\Delta,F)$ be a pushdown automaton and $\mathcal{C}\subseteq P(\Gamma\setminus\{\bot\})^*\{\bot\}$ be a regular set of configurations.
	Then the set
	\[\textsc{post}_\mathcal{P}^*(\mathcal{C}) := \{\mathfrak{c}\in P(\Gamma\setminus\{\bot\})^*\{\bot\} \mid \exists \mathfrak{d}\in\mathcal{C},u\in\Sigma^*: \mathcal{P}: \mathfrak{d}\xrightarrow{u}\mathfrak{c} \}\]
	of reachable configurations from $\mathcal{C}$ is regular.
	Moreover, an automaton defining $\textsc{post}_\mathcal{P}^*(\mathcal{C})$ can be effectively computed in polynomial time given $\mathcal{P}$ and an automaton defining $\mathcal{C}$.
\end{proposition}

\begin{proof}[of Theorem \ref{theorem_rec_in_aut_visibly_decidable}]
	Let $\mathcal{P}=(P,\Sigma,\Gamma,p_0,\bot,\Delta,F)$ be the given deterministic visibly pushdown automaton.
	We construct a synchronous transducer accepting distinct pairs $(p,\alpha\beta)$, $(p,\alpha\beta')$ of configurations falsifying the condition of Lemma \ref{lemma_rec_in_aut_visibly_core_lemma}.
	That is,
	\begin{enumerate}
		\item Both $(p,\alpha\beta)$ and $(p,\alpha\beta')$ are reachable from $(p_0,\bot)$,\label{visibly_proof_1}
		\item $|\alpha| \geq |P|^3+1$ (the $|P|^3+1$ topmost stack symbols are equal) and both configurations have the same state component, and\label{visibly_proof_2}
		\item they are not $\mathcal{P}$-equivalent.\label{visibly_proof_3}
	\end{enumerate}
	It suffices to construct synchronous transducers in polynomial time in $\mathcal{P}$ verifying \ref{visibly_proof_1}, \ref{visibly_proof_2}, and \ref{visibly_proof_3}, respectively.
	Then the claim follows because the intersection of synchronous transducers is computable in polynomial time.
	Furthermore, the obtained transducer defines the empty relation $\emptyset$ if and only if $L(\mathcal{P})$ is regular due to Lemma \ref{lemma_rec_in_aut_visibly_core_lemma}.
	The emptiness problem for synchronous transducer is decidable in polynomial time in terms of a graph search.
	
	Let $\mathfrak{C} := \textsc{post}_\mathcal{P}^*(\{(p_0,\bot)\})$ be the set of reachable configurations.
	Due to Proposition \ref{proposition_rec_in_aut_visibly} an automaton defining $\mathfrak{C}$ is computable in polynomial time.
	Thus, a synchronous transducer defining $\mathfrak{C}\times\mathfrak{C}$ is effectively obtainable in polynomial time, too (take two copies of the automaton for $\mathfrak{C}$ and let them run in parallel).
	$\mathfrak{C}\times\mathfrak{C}$ contains exactly all pairs of configurations satisfying \ref{visibly_proof_1}.
	Constructing a synchronous transducer verifying \ref{visibly_proof_2} is trivial.
	Its size is in $\mathcal{O}(|P|^3)$.
	
	It remains to construct a synchronous transducer verifying \ref{visibly_proof_3}.
	The idea is to guess a separating word and simulate $\mathcal{P}$ in parallel starting in the two configurations given as input to $\mathcal{A}$.
	For that purpose, it will be crucial to show that it suffices to guess only the return symbols of a separating word which are responsible for popping symbols from the stacks (instead of the whole separating word).
	
	By definition, two configurations $(p,\alpha\beta)$, $(p,\alpha\beta')$ are not $\mathcal{P}$-equivalent if and only if there is a word $z\in L(p,\alpha\beta)\triangle L(p,\alpha\beta')$ separating the configurations.
	Moreover, a separating word $z$ can be decomposed into $z=w_1r_1w_2r_2\ldots w_mr_mz'$ where the $w_i$ are well-matched words, the $r_i\in\Sigma_r$ are return symbols and $z'$ does not contain an unmatched return symbol (\textit{i.e.}\ $z'$'s structure is similar to a well-matched word but may contain additional call symbols).
	Note that the return symbols $r_i$ are the only symbols in $z$ allowing $\mathcal{P}$ to access the given stack contents $\alpha\beta$ and $\alpha\beta'$, respectively.
	Furthermore, it holds that $m\geq |\alpha|$. Otherwise, $z$ can certainly not separate the given configurations.
	On the other hand, $m\leq |\alpha\beta|$ or $m\leq |\alpha\beta'|$ does not hold necessarily.
	Indeed, $\mathcal{P}$ may pop the empty stack while processing $z$.
	We implement a nondeterministic synchronous transducer $\mathcal{A}$ that guesses $z$ and verifies that it separates the given configurations.
	For that purpose, it is only necessary to consider the return symbols $r_i$ in combination with the input.
	In particular, it is not necessary to simulate $\mathcal{P}$ step by step on the infixes $w_i$.
	
	The transducer $\mathcal{A}$ will maintain a pair of states $(q,s)$ of $\mathcal{P}$.
	Intuitively, the states $q$ and $s$ occur in runs of $\mathcal{P}$ starting in $(p,\alpha\beta)$ and $(p,\alpha\beta')$ on a separating word.
	Furthermore, $\mathcal{A}$ may proceed from $(q,s)$ to $(q',s')$ if there is a well-matched word $w$ and return symbol $r$ such that $\mathcal{P}$ can proceed from $q$ to $q'$ and $s$ to $s'$ via $wr$ and the topmost stack symbol, respectively.
	Note that in contrast to a full simulation, states in the run of $\mathcal{P}$ are skipped --- \textit{i.e.}\ precisely those states occurring in the run fragment on a well-matched word.
	The first pair of states is given by the input configurations (here $(p,p)$).
	Since $\mathcal{P}$'s behavior on the well-matched words $w_i$ is invariant under the stack contents which are the input of $\mathcal{A}$, the simulation of $\mathcal{P}$ on $w$ boils down to a reachability analysis of configurations.
	Moreover, the reachability analysis can be done at construction time.
	Recall that a synchronous transducer have to satisfy the property that no transition labeled $(a',b')$, $b'\neq \varepsilon$ can be taken after a transition labeled $(a,\varepsilon)$; the same applies to the first component.
	Therefore, $\mathcal{A}$ needs two control bits to handle the cases where $|\beta|\neq|\beta'|$ and $\mathcal{P}$ is popping the empty stack (which has to be done by $\varepsilon$-transitions of $\mathcal{A}$).
	That is,
	\[\mathcal{A} = ((\{q_0,q_f\}\cup P\times P)\times\{0,1\}^2,P\cup\Gamma,(q_0,0,0),\Delta_\mathcal{A},\{q_f\}\times\{0,1\}^2).\]
	The two control bits in the state space shall indicate that a transition of the form $(a,\varepsilon)$ or $(\varepsilon,a)$, respectively, has already been used.
	The accepting states --- \textit{i.e.}\ the first component is $q_f$ --- are used to indicate that the transducer guessed the postfix $z'$ of $z$ which does not contain unmatched returns.
	Afterwards, $\mathcal{A}$ must not simulate $\mathcal{P}$ any further --- hence, the accepting states are effectively sink states.
	Recall that for terms $t_1,t_2$ the indicator function defined by $\delta(t_1=t_2)$ evaluates to $1$ if $t_1=t_2$ and to $0$, otherwise.
	Given two valuation $i,j$ of the two control bits and $\mu,\nu\in\Gamma\cup\{\varepsilon\}$ we use the following shorthand notation to set the values $i',j'$ of the control bits in the next state:\\[.5em]
	$\textsc{valid}(i,j,\mu,\nu,i',j')$ holds if and only if
		\begin{enumerate}
			\item $i' = \delta(\mu = \varepsilon)$,
			\item $j' = \delta(\nu = \varepsilon)$,
			\item if $i = 1$ then $\mu = \varepsilon$, and
			\item if $j = 1$ then $\nu=\varepsilon$.
		\end{enumerate}
	Note that once a control bit is set to $1$ it cannot be reset to $0$.
	The transition relation of $\mathcal{A}$ is the union $\Delta_\mathcal{A} := \Delta_\text{aux}\cup\Delta_r\cup\Delta_{z'}$.
	The sets $\Delta_\text{aux}, \Delta_r$, and $\Delta_{z'}$ are defined as follows.
	\begin{multline*}
	\Delta_\text{aux}:= \{((q_0,0,0),(p,p), (p,p,0,0))\mid p\in P \}\\\cup \{((q_f,i,j),(\mu,\nu),(q_f,i',j'))\mid \mu,\nu\in\Gamma\cup\{\varepsilon\}, \textsc{valid}(i,j,\mu,\nu,i',j') \}
	\end{multline*}
	Starting in $q_0$ the transducer initializes the states of $\mathcal{P}$. Furthermore, once it is in the state $q_f$ the remaining input can be read. Recall that the guessed word $z$ may not pop the whole stacks of the configurations.
	Hence, it may be necessary to skip the remaining input. 
	The \underline{main} transitions guess a pair $w,r$ to pop a symbol from the stack:
	\begin{multline*}
	\Delta_r := \{((p,q,i,j),\mu,\nu,(p',q',i',j'))\mid\mu,\nu\in\Gamma\cup\{\varepsilon\}, \textsc{valid}(i,j,\mu,\nu,i',j'),\\ \exists r\in\Sigma_r,w\in\Sigma^*\text{ well-matched}:
	\mathcal{P}:(p,\mu\bot)\xrightarrow{wr}(p',\bot), (q,\nu\bot)\xrightarrow{wr}(q',\bot) \}.
	\end{multline*}
	
	Finally, the transducer can guess the trailing part $z'$ of $z$ which has no unmatched returns.
	Since it is the last part of the runs of $\mathcal{P}$ and $z$ separates the given configurations it has to lead to states $p',q'$ with $p'\in F \Leftrightarrow q'\notin F$.
	Note that $z' = \varepsilon$ and $p'=p,q'=q$ is a valid choice.
	Thus, there is no need to introduce transitions in $\Delta_r$ leading to accepting states.
	\begin{multline*}
	\Delta_{z'} := 
	\{((p,q,i,j),\mu,\nu,(q_f,i',j'))\mid \mu,\nu\in\Gamma\cup\{\varepsilon\}, \textsc{valid}(i,j,\mu,\nu,i',j'),\\ \exists p',q'\in P~\exists \lambda,\rho\in (\Gamma\setminus\{\bot\})^*\{\bot\}~\exists z'\in\Sigma^*: z'\text{ has no unmatched returns, and}\\
	p'\in F \Leftrightarrow q'\notin F, \mathcal{P}:(p,\bot)\xrightarrow{z'}(p',\lambda), \mathcal{P}:(q,\bot)\xrightarrow{z'}(q',\rho) \}.
	\end{multline*}
	
	The correctness follows immediately from  the fact that $\mathcal{P}$'s behavior on the $w_i$ as well as $z'$ is invariant under the stack content and the decomposition $z=w_1r_1\ldots w_mr_mz'$.
	Indeed, the input $(p\alpha\beta,p\alpha\beta')$ is accepted by $\mathcal{A}$ if and only if there is a word $z=w_1r_1\ldots w_mr_mz'$ such that $\mathcal{P}: (p,\alpha\beta)\xrightarrow{z}(p',\lambda)$ and $\mathcal{P}: (p,\alpha\beta')\xrightarrow{z}(q',\rho)$ for some stack contents $\lambda,\rho$ and $(p',q')\in F\times (P\setminus F)\cup (P\setminus F)\times F$ if and only if $(p,\alpha\beta)\not\approx_\mathcal{P}(p,\alpha\beta')$.
	
	Clearly, $\mathcal{A}$ has size polynomial in $\mathcal{P}$ but we have to show that the transition relation can be computed in polynomial time.
	For that purpose we consider the visibly pushdown automaton \[\mathcal{P}^2 := (P\times P,\Sigma,\Gamma\times\Gamma,(p_0,p_0),(\bot,\bot),\Delta^2,F\times F)\] where
	\begin{align*}
	\Delta^2 := ~&\{ ((p,q),c,(p',q'),(\mu,\nu))\mid (p,c,p',\mu),(q,c,q',\nu) \in\Delta \}\\
	~& \{((p,q),r,(\mu,\nu),(p',q'))\mid (p,r,\mu,p'),(q,r,\nu,q') \in\Delta \wedge \mu,\nu\neq\bot \}\\
	~& \{((p,q),a,(p',q'))\mid (p,a,p'),(q,a,q') \in\Delta \}.
	\end{align*}
	
	Informally, $\mathcal{P}^2$ simulates two copies of $\mathcal{P}$ on the same input.
	Note that we forbid to pop the empty stack by purging the respective transitions.
	Also, no conflict arises while using the stack because the pop and push behavior is controlled by the common input word.
	$\mathcal{P}$ can proceed from $(p,\zeta)$ to $(p',\zeta)$ via a well-matched word $w$ if and only if it can proceed from $(p,\bot)$ to $(p',\bot)$ without popping the empty stack.
	Since $\mathcal{P}^2$ cannot pop the empty stack, $\mathcal{P}$ can proceed from $p$ to $p'$ and from $q$ to $q'$ via a well-matched word $w$ if and only if the configuration $((p',q'),(\bot,\bot))$ is reachable from $((p,q),(\bot,\bot))$ by $\mathcal{P}^2$.
	The set of all these configurations can be determined by checking whether \[(p',q')(\bot,\bot)\in \textsc{post}^*_{\mathcal{P}^2}(\{(p,q)(\bot,\bot)\})\]
	holds.
	In turn, an automaton defining $\textsc{post}^*_{\mathcal{P}^2}(\{(p,q)(\bot,\bot)\})$ can be computed in polynomial time for each pair $(p,q)$ due to Proposition \ref{proposition_rec_in_aut_visibly}.
	Also, there are only $|P|^4$ many possible values for $p,q,p',q'$.
	Moreover, $\mathcal{P}:(p,\mu\bot)\xrightarrow{wr}(p'',\bot)$ holds for a well-matched word $w$ and $r\in\Sigma_r$ if and only if \[\mathcal{P}:(p,\mu\bot)\xrightarrow{w}(p',\mu\bot)\xrightarrow{r}(p'',\bot)\]
	for any $\mu\in\Gamma\cup\{\varepsilon\}$ and $r\in\Sigma_r$.
	Again there are only polynomial many combinations (in $|P|,|\Sigma|$ and $|\Gamma|$).
	Altogether, we conclude that $\Delta_r$ can be effectively obtained in polynomial time.
	The transition set $\Delta_{z'}$ can be computed similarly.
	Since the guessed words $z'$ are not well-matched but do not touch the existing stack content  (they do not have unmatched returns), it has to be verified whether
	\[(p',q')\zeta\in \textsc{post}^*_{\mathcal{P}^2}(\{(p,q)(\bot,\bot)\})\text{ for some } \zeta\in\smash{\big((\Gamma\setminus\{\bot\})\{\bot\}\big)^2}.\]
	This is achievable in polynomial time by a graph search because Proposition \ref{proposition_rec_in_aut_visibly} provides an automaton defining $\textsc{post}^*_{\mathcal{P}^2}(\{(p,q)(\bot,\bot)\})$ of polynomial size for each pair of states $(p,q)$.
\end{proof}

\subsection{Deciding Recognizability of Binary Automatic Relations}
With the regularity test for \textsc{DVPA}s established we turn towards our second objective which is to decide recognizability of binary automatic relations.
Recall that for a word $u$ we denote its reversal by $\rev(u)$.
 
\begin{lemma}\label{lemma_rec_in_aut_visibly_construction}
	Let $R\subseteq \Sigma_1^*\times\Sigma_2^*$ with $\Sigma_1\cap\Sigma_2=\emptyset$ be an automatic relation and $\#\notin\Sigma_1\cup\Sigma_2$ be a fresh symbol.
	Furthermore, let $\mathcal{A}$ be a (nondeterministic) synchronous transducer defining $R$.
	Then $L_R := \{\rev(u)\#v\mid (u,v)\in R \}$ is definable by a \textsc{DVPA} whose size is single exponential in $|\mathcal{A}|$.
\end{lemma}
\begin{proof}
	Let $\mathcal{A} = (Q,\Sigma_1,\Sigma_2,q_0,\Delta,F)$ be the given synchronous transducer.
	W.l.o.g.\ we assert that $\mathcal{A}$ does not have any transitions labeled $(\varepsilon,\varepsilon)$.
	Otherwise, they can be eliminated in polynomial time using the well-known standard $\varepsilon$-elimination procedure for $\varepsilon$-automata.
	The basic idea is to push $\rev(u)$ to the stack and use the stack as the read-only input tape to simulate $\mathcal{A}$ on $(u,v)$. For that purpose, $\Sigma_1$ becomes the set of call symbols to push $\rev(u)$ to the stack and $\Sigma_2$ becomes the set of return symbols to be able to read letters of $u$ and $v$ simultaneously.
	Unfortunately, if $\rev(u)$ is longer than $v$ then the pushdown automaton is not able to simulate transitions labeled $(a,\varepsilon)$ because there are no return symbols left.
	In other words, $u$ cannot be read to the end if $u$ is longer than $v$.
	To solve this problem the pushdown automaton performs a reverse powerset construction on $\rev(u)$ while pushing it to the stack using only transitions of the form $(p,a,\varepsilon,q)$ and stores the states on the stack --- \textit{i.e.}\ it starts with the set of accepting states and computes the set of states from which the current set of states is reachable by a transition labeled $(a,\varepsilon)$ where $a$ is supposed to be pushed to the stack.
	That way it knows whether the transducer $\mathcal{A}$ could proceed to an accepting state using the remaining part of $\rev(u)$ and transitions labeled $(a,\varepsilon)$ once $v$ has been read completely.
	Note that a reverse powerset construction yields an exponential blow-up even for deterministic transducers. Therefore, it is pointless to determinize $\mathcal{A}$ first.
	Instead, the "normal" forward powerset construction is incorporated into the construction such that the resulting visibly pushdown automaton is deterministic.\pagebreak
	
	Formally, let $\mathcal{P} := (2^Q~\dot{\cup}~(2^Q\times 2^Q),\Sigma,\Gamma,F,(\varepsilon,F),\Delta_\mathcal{P},F_\mathcal{P})$ where
	\begin{itemize}
		\item $\Sigma = \Sigma_c~\dot{\cup}~\Sigma_r~\dot{\cup}~\Sigma_\text{int}$ with $\Sigma_c := \Sigma_1$, $\Sigma_r := \Sigma_2$ and $\Sigma_\text{int} := \{\#\}$,
		\item $\Gamma := (\Sigma_c\cup\{\varepsilon\})\times 2^Q$,
		\item $F_\mathcal{P} := \{(P,S)\in 2^Q\times 2^Q\mid P\cap S\neq \emptyset \}$,
	\end{itemize}
	The states in $2^Q$ are used while pushing $\rev(u)$ to the stack and perform the reverse powerset construction. Similarly, states in $2^Q\times 2^Q$ are used to recover the constructed subsets of the reverse powerset construction (second component) and to perform the "normal" powerset construction (first component).
	Furthermore, note that the bottom stack symbol is $(\varepsilon,F)$ indicating that $\rev(u)$ has been read completely and that the transducer should be in an accepting state.
	Further on, the transition relation of $\mathcal{P}$ is defined by
	\begin{align*}
	\Delta_\mathcal{P} := ~&\{(S,c,S',(c,S)) \mid c\in\Sigma_c,S\in 2^Q, S'=\{s' \mid \exists s\in S: (s',c,\varepsilon,s)\in\Delta\} \}\\
	\cup~&\{(P,\#,(\{q_0\}, P))\mid P\in 2^Q \}\\
	\cup~&\{((P,S),r,(c,S'), (P',S'))\mid P,S,S'\in 2^Q, r\in\Sigma_r, c\in\Sigma_c\cup\{\varepsilon\}, \\
	~&\quad\quad\quad\quad\quad\quad\quad\quad\quad\quad\quad\quad\quad\quad P' = \{p'\mid \exists p\in P: (p,c,r,p')\in\Delta \} \}
	\end{align*}
	It is easy to see that $\mathcal{P}$ is deterministic.
	Moreover, for $u=a_1\ldots a_n$ the stack content has the form $(a_1,S_1)\ldots(a_n,S_n)(\varepsilon,F)$ when $\mathcal{P}$ reads the $\#$-symbol.
	Furthermore, it holds that $\mathcal{A}:s_i\xrightarrow{\smash{(a_{i+1}\ldots a_n,\varepsilon)}}F$ for precisely each $s_i\in S_i$, $1\leq i\leq n$.
	In other words, $\mathcal{A}$ accepts the remaining part of $u$ on the stack precisely from all states in $S_i$.
	In particular, this claim holds for the case $u=\varepsilon$ by the choice of the start state $F$.
	Suppose $|v|\geq |u|$.
	Then the pushdown automaton $\mathcal{P}$ simulates $\mathcal{A}$ on $(u,v)$ and ends up in a state $(P,F)$ where $P$ is the set of states reachable by $\mathcal{A}$ given the input $(u,v)$ (powerset construction). Note that the second component is $F$ because the stack has been cleared and the start state is $F$ (in the case $|v| = |u|$) and the bottom stack symbol is $(\varepsilon,F)$ (in the case $|v| > |u|$).
	Thus, $\mathcal{P}$ accepts if and only if $P\cap F\neq \emptyset$.
	It follows that $\mathcal{P}$ accepts $\rev(u)\#v$ if and only if $\mathcal{A}$ accepts $(u,v)$ in the case $|v|\geq |u|$.
	If $|v| < |u|$ then $\mathcal{P}$ ends up in a state $(P,S)$. Let $m:=|v|$.
	Then $P$ is the set of states reachable by $\mathcal{A}$ given the input $(a_1\ldots a_m,v)$ analogously to the case $|v|\geq |u|$. Furthermore, by our previous observation $S$ is the set of states from which $\mathcal{A}$ accepts $(a_{m+1}\ldots,a_n,\varepsilon)$.
	Thus, $\mathcal{P}$ accepts $\rev(u)\#v$ if and only if $P\cap S\neq\emptyset$ if and only if there is an accepting run of $\mathcal{A}$ on $(u,v)$.
\end{proof}

Since $L_R$ is regular if and only if $R$ is a recognizable relation as shown by \cite{CCG06}, we obtain the second result of this section as corollary of Theorem \ref{theorem_rec_in_aut_visibly_decidable} and Lemma \ref{lemma_rec_in_aut_visibly_construction}.
\begin{corollary}\label{corollary_rec_in_aut_decidable_visibly}
	Let $\mathcal{A}$ be a (possibly nondeterministic) synchronous transducer defining a binary relation.
	Then it is decidable in single exponential time whether $R_*(\mathcal{A})$ is recognizable.
\end{corollary}
\begin{proof}
	W.l.o.g.\ $\mathcal{A}$ defines a binary relation $R\subseteq \Sigma_1\times\Sigma_2$ where $\Sigma_1\cap\Sigma_2=\emptyset$. Otherwise, alphabet symbols can easily be renamed in one of the components.
	Due to Lemma \ref{lemma_rec_in_aut_visibly_construction} we can obtain a deterministic visibly pushdown automaton defining $L_R=\{\rev(u)\#v\mid (u,v)\in R \}$ in single exponential time.
	Furthermore, by Theorem \ref{theorem_rec_in_aut_visibly_decidable} it can be decided in polynomial time whether $L_R$ is regular.
	Hence, it suffices to show that $L_R$ is regular if and only if $R$ is recognizable.
	Then the claim follows immediately.
	
	Suppose $R$ is recognizable. Then $R$ can be written as $R=\bigcup_{i=1}^m L_i\times K_i$ for some regular languages $L_i$ and $K_i$.
	Thus, $L_R = \bigcup_{i=1}^m\rev(L_i)\{\#\}K_i$ is regular.
	On the contrary, assume that $L_R$ is regular.
	Then there is a finite automaton $\mathcal{B} = (Q,\Sigma_1\cup\Sigma_2\cup\{\#\},q_0,\Delta,F)$ defining $L_R$.
	Let $L_{q_0q} := \{w\mid \mathcal{B}:q_0\xrightarrow{w}q \}$ and  $L_{qF} := \{w\mid \mathcal{B}:q\xrightarrow{w}F \}$ for $q\in Q$.
	Then we can write $L_R$ as $L_R = \bigcup_{(p,\#,q)\in\Delta} L_{q_0p}\{\#\}L_{qF}$.
	Furthermore, for all these $L_{q_0p},L_{qF}\neq\emptyset$ (where $(p,\#,q)\in\Delta$) we have that $L_{q_0p}\subseteq\Sigma_1^*$ and $L_{qF}\subseteq\Sigma_2^*$.
	It follows that $R=\bigcup_{(p,\#,q)\in\Delta} \rev(L_{q_0p})\times L_{qF}$ is recognizable.
\end{proof}
 \section{Conclusion} \label{sec:conclusion}
The undecidability of the equivalence problem for deterministic
$\omega$-rational relations presented in Section~\ref{sec:equiv-drat}
exhibits an interesting difference between deterministic transducers
on finite and on infinite words. We believe that it is worth to
further study the algorithmic theory of this class of relations. For
example, the decidability of recognizability for a given
deterministic $\omega$-rational relation is an open question. The
technique based on the connection between binary rational relations
and context-free languages as presented in
Section~\ref{sec:rec-in-aut} that is used by \cite{CCG06} for deciding
recognizability of deterministic rational relations cannot be
(directly) adapted. First of all, the idea of pushing the first
component on the stack and then simulating the transducer while reading
the second component fails because this would require an infinite
stack. Furthermore, the regularity problem for deterministic
$\omega$-pushdown automata is not known to be decidable (only for the
subclass of deterministic weak B\"uchi automata \cite{LodingR12} were able to show decidability).

It would also be interesting to understand whether the decidability of
the synthesis problem (see the introduction) for deterministic
rational relations over finite
words recently proved by \cite{FiliotJLW16} can be transferred to infinite words.

For the recognizability problem of \maybeomega{}automatic relations
we have shown decidability with a doubly exponential time algorithm for
infinite words. We also provided a singly exponential time algorithm for
the binary case over finite words (improving the complexity of the
approach of \cite{CCG06} as explained in
Section~\ref{sec:omega-rec-in-aut}).
It remains open whether the singly exponential time algorithm can be extended to automatic relations of arbitrary arity.
Also, it is open whether there are matching lower complexity bounds.

The connection between automatic relations and \textsc{VPA}s raises the question, whether extensions of \textsc{VPA}s studied in the literature (as for example by
\cite{Caucal06}) can be used to identify interesting subclasses of
relations between the \maybeomega{}automatic and deterministic
\maybeomega{}rational relations. The problem of identifying such classes
for the case of infinite words has already been posed by \cite{Thomas92}. 
\acknowledgements We thank the reviewers of the DMTCS journal for their remarks which led to improvements of the presentation, in particular of the proof of Lemma~\ref{rec_in_aut_omega_lemma_1}.

\bibliographystyle{abbrvnat}

\end{document}